\documentclass[11pt]{article}
\usepackage[margin=1.25in,dvips]{geometry}
\usepackage{amsthm}
\usepackage{fullpage}

\usepackage{times}
\usepackage{graphicx}
\usepackage{amssymb}
\usepackage{amsmath}
\usepackage{float}
\usepackage{url}
\usepackage{enumerate}
\usepackage{subfig}
\usepackage[boxed,section]{algorithm}
\usepackage{algpseudocode}
\usepackage{color}
\usepackage[table]{xcolor}

\usepackage{epsfig} 
\usepackage{graphics}
\usepackage{amsfonts}
\usepackage{amssymb}
\usepackage{amsmath}
\usepackage{algorithm}
\usepackage{algorithmicx}
\usepackage[fleqn,tbtags]{mathtools}
\usepackage{amsthm}
\usepackage{tikz}
\usetikzlibrary{arrows}

\newcommand{\footremember}[2]{%
   \footnote{#2}
    \newcounter{#1}
    \setcounter{#1}{\value{footnote}}%
}
\newcommand{\footrecall}[1]{%
    \footnotemark[\value{#1}]%
} 

\newtheorem{theorem}{Theorem}[section]
\newtheorem{lemma}[theorem]{Lemma}
\newtheorem{corollary}[theorem]{Corollary}
\newtheorem{definition}[theorem]{Definition}

\newtheorem{remark}[theorem]{Remark}

\newtheorem{notation}[theorem]{Notation}

\newcommand{\comments}[1]{}
\newcommand{\defeq}{\mathrel{\mathop:}=}

\newcommand{\hide}[1]{}

{\makeatletter
 \gdef\xxxmark{%
   \expandafter\ifx\csname @mpargs\endcsname\relax 
     \expandafter\ifx\csname @captype\endcsname\relax 
       \marginpar{xxx}
     \else
       xxx 
     \fi
   \else
     xxx 
   \fi}
 \gdef\xxx{\@ifnextchar[\xxx@lab\xxx@nolab}
 \long\gdef\xxx@lab[#1]#2{{\bf [\xxxmark #2 ---{\sc #1}]}}
 \long\gdef\xxx@nolab#1{{\bf [\xxxmark #1]}}
}

{\begin{itemize}%
\setlength{\itemsep}{0pt}%
\setlength{\topsep}{0pt}%
\setlength{\partopsep}{0pt}%
\setlength{\parskip}{0pt}}%
{\end{itemize}}
{\begin{enumerate}%
\setlength{\itemsep}{0pt}%
\setlength{\topsep}{0pt}%
\setlength{\partopsep}{0pt}%
\setlength{\parskip}{0pt}}%
{\end{enumerate}}
\newfont{\aufnt}{phvr8t at 12pt}
\newfont{\affaddr}{phvr8t at 10pt}

\makeatletter
\newtheorem*{rep@theorem}{\rep@title}
\newcommand{\newreptheorem}[2]{%
\newenvironment{rep#1}[1]{%
 \def\rep@title{#2 \ref{##1}}%
 \begin{rep@theorem}}%
 {\end{rep@theorem}}}
\makeatother

\newreptheorem{theorem}{Theorem}
\newreptheorem{lemma}{Lemma}

\begin{document}


\title{Linear Programming Decoding of Spatially Coupled Codes}
\author{%
    Louay Bazzi\footremember{alley}{Department of Electrical and Computer Engineering, American University of Beirut, Beirut, Lebanon.}%
    \and Badih Ghazi\footrecall{alley}%
    \and R\"udiger Urbanke\footremember{trailer}{ECE Department, EPFL, Lausanne, Switzerland.}%
}
\date{{\tt lb13@aub.edu.lb, bbg01@mail.aub.edu, ruediger.urbanke@epfl.ch}}
\maketitle
\begin{abstract}
For a given family of spatially coupled codes, we prove that the LP threshold on the BSC of the graph cover ensemble is the same as the LP threshold on the BSC of the derived spatially coupled ensemble. This result is in contrast with the fact that the BP threshold of the derived spatially coupled ensemble is believed to be larger than the BP threshold of the graph cover ensemble \cite{kudekar2011threshold}, \cite{kudekar2012spatially}. To prove this, we establish some properties related to the dual witness for LP decoding which was introduced by \cite{feldman2007lp} and simplified by \cite{daskalakis2008probabilistic}. More precisely, we prove that the existence of a dual witness which was previously known to be sufficient for LP decoding success is also necessary and is equivalent to the existence of certain acyclic hyperflows. We also derive a sublinear (in the block length) upper bound on the weight of any edge in such hyperflows, both for regular LPDC codes and for spatially coupled codes and we prove that the bound is asymptotically tight for regular LDPC codes. Moreover, we show how to trade crossover probability for ``LP excess'' on all the variable nodes, for any binary linear code.
\end{abstract}





\tableofcontents

\section[Introduction]{\Large{\bf Introduction}}

\subsection{Binary linear codes}\label{le:binarylinearcodes}
A binary linear code $\zeta$ of block length $n$ is a subspace of the $\mathbb{F}_{2}$-vector space $\mathbb{F}_{2}^n$. The $\epsilon$-BSC (Binary Symmetric Channel) with input $X \in \mathbb{F}_{2}^n$ and output $Y \in \mathbb{F}_{2}^n$ flips each input bit independently with probability $\epsilon$. Let $\gamma$ be the log-likelihood ratio vector which is given by $\gamma_{i}=\log{\big(\frac{p_{Y_{i}|X_{i}}(y_{i}|0)} {p_{Y_{i}|X_{i}}(y_{i}|1)}\big)} = (-1)^{y_i}\log{\frac{1-\epsilon}{\epsilon}}$ for any $i \in \{1,\dots,n\}$. The optimal decoder is the Maximum Likelihood (ML) decoder which is given by
\begin{align*}
\hat{x}_{ML} &= \underset{x \in \zeta}{\operatorname{argmax}}{~ p_{Y|X}(y|x)} = \underset{x \in \zeta}{\operatorname{argmax}}{~ \prod_{i=1}^{n} p_{Y_{i}|X_{i}}(y_{i}|x_{i})} = \underset{x \in \zeta}{\operatorname{argmax}}{~ \frac{ \prod_{i=1}^{n} p_{Y_{i}|X_{i}}(y_{i}|x_{i})}  { \prod_{i=1}^{n} p_{Y_{i}|X_{i}}(y_{i}|0)} }\\ 
&= \underset{x \in \zeta}{\operatorname{argmax}}{~ \log{\bigg(\prod_{i=1}^{n} \frac{p_{Y_{i}|X_{i}}(y_{i}|x_{i})} {p_{Y_{i}|X_{i}}(y_{i}|0)}\bigg)}} = \underset{x \in \zeta}{\operatorname{argmax}}{~ \displaystyle\sum\limits_{i=1}^{n} \log{\bigg(\frac{p_{Y_{i}|X_{i}}(y_{i}|x_{i})} {p_{Y_{i}|X_{i}}(y_{i}|0)}\bigg)}} = \underset{x \in \zeta}{\operatorname{argmin}}{~ \displaystyle\sum\limits_{i=1}^{n} \gamma_{i}x_{i}}
\end{align*}
where the second equality follows from the fact that the channel is memoryless. Since the objective function is linear in $x$, replacing $\zeta$ by the convex span $conv(\zeta)$ of $\zeta$\comments{\footnote{The convex hull of a subset $S$ of $\mathbb{R}^{n}$ is defined to be $conv(S) = \{ \alpha x + (1-\alpha)y ~ | ~ x,y \in S \text{ and } \alpha \in [0,1] \}$.}} does not change the value of the minimal solution. Hence, we get
\begin{equation}\label{le:ML_dec_LP}
\hat{x}_{ML} = \underset{x \in conv(\zeta)}{\operatorname{argmin}}{~ \displaystyle\sum\limits_{i=1}^{n} \gamma_{i}x_{i}}
\end{equation}
ML decoding is known to be NP-hard for general binary linear codes \cite{berlekamp1978inherent}. This motivates the study of suboptimal decoding algorithms that have small running times.

\subsection{Linear programming decoding}
LP (Linear Programming) decoding was introduced by \cite{feldman2005using} and is based on the idea of replacing $conv(\zeta)$ in (\ref{le:ML_dec_LP}) with a larger subset of $\mathbb{R}^n$, with the goal of reducing the running time while maintaining a good error correction performance. First, note that $conv(\zeta) = conv(\underset{j \in C} \bigcap \zeta_{j})$ where $\zeta_{j} = \{ z \in \{0,1\}^{n}: \text{ $w(z|_{N(j)})$ is even}\}$\footnote{For $x \in \{0,1\}^n$ and $S \subseteq \{1,\dots,n\}$, $x|_S \in \{0,1\}^{n}$ denotes the restriction of $x$ to $S$ i.e. $(x|_S)_i = x_i$ if $i \in S$ and $(x|_S)_i = 0$ otherwise, and $w(x)$ denotes the Hamming weight of $x$.} for all $j$ in the set $C$ of check nodes corresponding to a fixed Tanner graph of $\zeta$ and where $N(j)$ is the set of all neighbors of check node $j$. Then, LP decoding is given by relaxing $conv(\underset{j \in C} \bigcap \zeta_{j})$ to $\underset{j \in C} \bigcap conv(\zeta_{j})$:
\begin{equation}\label{le:relaxed_LP}
\hat{x}_{LP} = \underset{x \in P}{\operatorname{argmin}}{~ \displaystyle\sum\limits_{i=1}^{n} \gamma_{i}x_{i}}
\end{equation}
where $P = \underset{j \in C} \bigcap conv(\zeta_{j})$ is the so-called ``fundamental polytope'' that will be carefully considered in the proof of Theorem \ref{le:existencedualwitness}. A central property of $P$ is that it can be described by a linear number of inequalities, which means that the linear program (\ref{le:relaxed_LP}) can be solved in time polynomial in $n$ using the ellipsoid algorithm or interior point methods.\\ 
When analyzing the operation of LP decoding, one can assume that the all-zeros codeword was transmitted \cite{feldman2005using}. Then, by normalizing the expression for the log-likelihood ratio $\gamma$ given in Section \ref{le:binarylinearcodes} by the positive constant $\log(\frac{1-\epsilon}{\epsilon})$, we can assume that the log-likelihood ratio is given by $\gamma_{i}=1$ if $y_{i} = 0$ and $\gamma_{i}=-1$ if $y_{i} = 1$ for all $i \in \{1,\dots,n\}$. As in previous work, we make the conservative assumption that LP decoding fails whenever there are multiple optimal solutions to the linear program (\ref{le:relaxed_LP}). In other words, under the all zeros assumption, LP decoding succeeds if and only if the zero codeword is the unique optimal solution to the linear program (\ref{le:relaxed_LP}). In order to show that LP decoding corrects a constant fraction of errors when the Tanner graph has sufficient expansion, \cite{feldman2007lp} introduced the concept of a dual witness, which is a dual feasible solution with zero cost and with a given set of constraints having a positive slack. By complementary slackness, it follows that the existence of a dual witness implies LP decoding success \cite{feldman2007lp}. A simplified (but equivalent) version of this dual witness, called a hyperflow, was introduced in \cite{daskalakis2008probabilistic} (and later generalized in \cite{halabi2012linear}) and used to prove that LP decoding can correct a larger fraction of errors in a probabilistic setting. This hyperflow will be described in Section \ref{le:LPandgraphs}. However, it was unkown whether the existence of a hyperflow (or equivalently that of a dual witness) is necessary for LP decoding success. We will show, by careful consideration of the fundamental polytope $P$, that this is indeed the case.

\subsection{Spatially coupled codes}

The idea of spatial coupling has been recently used in coding theory, compressive sensing and other fields. Spatially coupled codes (or convolutional LDPC codes) were introduced in \cite{jimenez1999time}. Recently, \cite{kudekar2011threshold} showed that the BP threshold of spatially coupled codes is the same as the MAP (Maximum Aposteriori Probability) threshold of the base LDPC code in the case of the Binary Erasure Channel (BEC). Moreover, \cite{kudekar2012spatially} showed that spatially coupled codes achieve capacity under belief propagation. In compressive sensing, \cite{krzakala2012statistical} and \cite{donoho2012information} showed that spatial coupling can be used to design dense sensing matrices that achieve the same peformance as the optimal $l_{0}$-norm minimizing compressive sensing decoder. In coding theory, the intuition behind the improvement in performance due to spatial coupling is that the check nodes located at the boundaries have low degrees which enables the BP algorithm to initially recover the transmitted bits at the boundaries. Then, the other transmitted bits are progressively recovered from the boundaries to the center of the code. A similar intuition is behind the good performance of spatial coupling in compressive sensing \cite{donoho2012information}.

\subsection{The conjecture}
It was reported by \cite{Davbur} that, based on numerical simulations, spatial coupling does not seem to improve the performance of LP decoding. This lead to the conjecture that the LP threshold of a spatially coupled ensemble on the BSC is the same as that of the base ensemble. A natural approach to prove this claim is twofold:
\begin{enumerate}
\item Show that the LP threshold of the spatially coupled ensemble on the BSC is the same as that of the graph cover ensemble.
\item Show that the LP threshold of the graph cover ensemble on the BSC is the same as that of the base ensemble.
\end{enumerate}

\subsection{Contributions}
We prove the first part of the conjecture. To do so, we prove some general results about LP decoding of LDPC codes that may be of independent interest.
\begin{enumerate}
\item We prove that the existence of a dual witness which was previously known to be sufficient for LP decoding success is also necessary and is equivalent to the existence of certain acyclic hyperflows (Theorem \ref{le:existencedualwitness}).
\item We derive a sublinear (in the block length) upper bound on the weight of any edge in the hyperflow, for regular LDPC codes (Theorem \ref{le:maxweight}) and spatially coupled codes (Theorem \ref{le:maxweightsc}). In the regular case, we show that our bound is asymptotically tight (Theorem \ref{le:asymptotictightness}).
\item We show how to trade crossover probability for ``LP excess'' on all the variable nodes, for any binary linear code (Theorem \ref{le:interplaytheorem}).
\end{enumerate}
We leave the second part of the conjecture open.

\subsection{Outline}
The paper is organized as follows. In Section \ref{le:main_result_section}, we formally state the main result of the paper. In Section \ref{le:LPandgraphs}, we prove that the existence of a dual witness which was previously known to be sufficient for LP decoding success is also necessary and is equivalent to the existence of certain weighted directed acyclic graphs. In Section \ref{le:transformWDAG}, we show how to transform those weighted directed acyclic graphs into weighted directed forests while preserving their central properties. In Section \ref{le:maxweightregular}, we prove, using the result of Section \ref{le:transformWDAG}, a sublinear (in the block length) upper bound on the weight of any edge in such graphs, for regular codes. An analogous upper bound is proved in Section \ref{le:maxweightsc_section} for spatially coupled codes. In Section \ref{le:relation_section}, we relate LP decoding on a graph cover code and on a spatially coupled code. In Section \ref{le:interplay_section}, we show how to trade crossover probability for ``LP excess'' on all the variable nodes, for any binary linear code. The results of Sections \ref{le:maxweightsc_section}, \ref{le:relation_section} and \ref{le:interplay_section} are finally used in Section \ref{le:proof_main_result_section} where we prove the main result of the paper.

\subsection{Notation and terminology}
We denote the set of all non-negative integers by $\mathbb{N}$. For any integers $n,a,b$ with $n \geq 1$, we denote by $[n]$ the set $\{1, \dots ,n\}$ and by $[a:b]$ the set $\{a, \dots ,b\}$. For any event $A$, let $\overline{A}$ be the complement of $A$. For any vertex $v$ of a graph $G$, we let $N(v)$ denote the set of all neighbors of $v$ in $G$. For any $x \in \{0,1\}^n$ and any $S \subseteq [n]$, let $x|_S \in \{0,1\}^{n}$ s.t. $(x|_S)_i = x_i$ if $i \in S$ and $(x|_S)_i = 0$ otherwise. A binary linear code $\zeta$ can be fully described as the nullspace of a matrix $H \in \mathbb{F}_{2}^{(n-k) \times n}$, called the parity check matrix of $\zeta$. For a fixed $H$, $\zeta$ can be graphically represented by a Tanner graph $(V,C,E)$ which is a bipartite graph where $V=\{v_{1},\dots,v_{n}\}$ is the set of variable nodes, $C=\{c_{1},\dots,c_{n-k}\}$ is the set of check nodes and for any $i \in [n]$ and any $j \in [n-k], ~ (v_{i},c_{j}) \in E$ if and only if $H_{j,i}=1$. If $H$ is sparse, then $\zeta$ is called a Low Density Parity Check (LDPC) code. LDPC codes were introduced and first analyzed by Gallager \cite{gallager1962low}. If the number of ones in each column of $H$ is $d_{v}$ and the number of ones in each row of $H$ is $d_{c}$, $\zeta$ is called a $(d_{v},d_{c})$-regular code. We let $\hat{d_{v}}=(d_{v}-1)/2$. Throughout the paper, we assume that $n,d_{c},d_{v} > 2$.

\section[Main result]{\Large{\bf Main result}}\label{le:main_result_section}
First, we define the spatially coupled codes under consideration.
\begin{definition}\label{le:spatiallycoupledensemble} (Spatially coupled code)\\ 
A $(d_{v},d_{c}=kd_{v},L,M)$ spatially coupled code, with $d_{v}$ an odd integer and $M$ divisible by $k$, is constructed by considering the index set $[-L-\hat{d_{v}}:L+\hat{d_{v}}]$ and satisfying the following conditions:\footnote{Informally, $2L+1$ is the number of ``layers'' and $M$ is the number of variable nodes per ``layer''.}
\begin{enumerate}
\item $M$ variable nodes are placed at each position in $[-L:L]$ and $M \frac{d_{v}}{d_{c}}$ check nodes are placed at each position in $[-L-\hat{d_{v}}:L+\hat{d_{v}}]$.
\item For any $j \in [-L+\hat{d_{v}}:L-\hat{d_{v}}]$, a check node at position $j$ is connected to $k$ variable nodes at position $j+i$ for all $i \in [-\hat{d_{v}}:\hat{d_{v}}]$.
\item For any $j \in [-L-\hat{d_{v}}:-L+\hat{d_{v}}-1]$, a check node at position $j$ is connected to $k$ variable nodes at position $i$ for all $i \in [-L:j+\hat{d_{v}}]$.
\item For any $j \in [L-\hat{d_{v}}+1:L+\hat{d_{v}}]$, a check node at position $j$ is connected to $k$ variable nodes at position $i$ for all $i \in [j-\hat{d_{v}}:L]$.
\item\label{le:additional_constraint} No two check nodes at the same position are connected to the same variable node.
\end{enumerate}
\end{definition}

With the exception of the non-degeneracy condition \ref{le:additional_constraint}, Definition \ref{le:spatiallycoupledensemble} above is the same as that given in Section II-A of \cite{kudekar2011threshold}. We next define the graph cover codes under consideration which are similar to the tail-biting LDPC convolutional codes introduced by \cite{tavares2007tail}.
\begin{definition}\label{le:graphcoverensemble} (Graph cover code)\\ 
A $(d_{v},d_{c}=kd_{v},L,M)$ graph cover code, with $d_{v}$ an odd integer and $M$ divisible by $k$, is constructed by considering the index set $[-L:L]$ and satisfying the following conditions:
\begin{enumerate}
\item $M$ variable nodes and $M \frac{d_{v}}{d_{c}}$ check nodes are placed at each position in $[-L:L]$.
\item For any $j \in [-L:L]$, a check node at position $j$ is connected to $k$ variable nodes at position $(j+i) \mod [-L:L]$ for all $i \in [-\hat{d_{v}}:\hat{d_{v}}]$.
\item No two check nodes at the same position are connected to the same variable node.
\end{enumerate}
\end{definition}

Note that ``cutting'' a  graph cover code at any position $i \in [-L:L]$ yields a spatially coupled code. This motivates the following definition.
\begin{definition}\label{le:derivedpsatiallycoupledcodes} (Derived spatially coupled codes)\\ 
Let $\zeta$ be a $(d_{v},d_{c}=kd_{v},L,M)$ graph cover code. For each $i \in [-L:L]$, the $(d_{v},d_{c}=kd_{v},L-\hat{d_{v}},M)$ spatially coupled code $\zeta'_{i}$ is obtained from $\zeta$ by removing all $M$ variable nodes and their adjacent edges at each position $i+j \mod [-L:L]$ for every $j \in [0:2\hat{d_{v}}-1]$. Then, $\mathcal{D}(\zeta) = \{\zeta'_{-L}, \dots ,\zeta'_{L}\}$ is the set of all $2L+1$ derived spatially coupled codes of $\zeta$.
\end{definition}

\begin{definition}(Ensembles and Thresholds)\\ 
Let $\Gamma$ be an ensemble i.e a probability distribution over codes. The LP threshold $\xi$ of $\Gamma$ on the BSC is defined as $\xi = \sup \{\epsilon > 0 ~ | ~  Pr_{\zeta \sim \Gamma \atop \epsilon\text{-}BSC}[\text{LP error on }\zeta] = o(1)\}$.
\end{definition}

We are now ready to state the main result of this paper.
\begin{theorem}\label{le:equalitytheorem} (Main result: $\xi_{GC} = \xi_{SC}$)\\ 
Let $\Gamma_{GC}$ be a $(d_{v},d_{c}=kd_{v},L,M)$ graph cover ensemble with $d_{v}$ an odd integer and $M$ divisible by $k$. Let $\Gamma_{SC}$ be the $(d_{v},d_{c}=kd_{v},L-\hat{d_{v}},M)$ spatially coupled ensemble which is sampled by choosing a graph cover code $\zeta \sim \Gamma_{GC}$ and returning a element of $\mathcal{D}(\zeta)$ chosen uniformly at random\footnote{Here, $\mathcal{D}(\zeta)$ refers to Definition \ref{le:derivedpsatiallycoupledcodes}.}. Denote by $\xi_{GC}$ and $\xi_{SC}$ the respective LP threholds of $\Gamma_{GC}$ and $\Gamma_{SC}$ on the BSC. There exists $\nu>0$ depending only on $d_{v}$ and $d_{c}$ s.t. if $M = o(L^{\nu})$ and $\Gamma_{SC}$ satisfies the property that for any constant $\Delta > 0$, 
\begin{equation}\label{le:needed_condition}
Pr_{\zeta' \sim \Gamma_{SC} \atop (\xi_{SC} - \Delta)\text{-}BSC}[\text{LP error on }\zeta'] = o(\frac{1}{L^2})
\end{equation}
Then, $\xi_{GC} = \xi_{SC}$.
\end{theorem}

Note that for $M=\omega(\log{L})$, condition (\ref{le:needed_condition}) above is expected to hold for the spatially coupled ensemble $\Gamma_{SC}$ since under typical decoding algorithms, the error probability on the $(\xi_{SC} - \Delta)$-BSC is expected to decay to zero as $O(L e^{-c \times \Delta^{2} \times M})$ for some constant $c>0$. Moreover, note that in the regime $M=\Theta(L^{\delta})$ (for any positive constant $\delta$), spatial coupling provides empirical improvements under iterative decoding and in fact, the improvement is expected to take place as long as $L$ is subexponential in $M$ \cite{olmos2011scaling}. 

\section[LP decoding, dual witnesses, hyperflows and WDAGs]{\Large{\bf LP decoding, dual witnesses, hyperflows and WDAGs}}\label{le:LPandgraphs}

The following definition is based on Definition 1 of \cite{feldman2007lp}.

\begin{definition}\label{le:dualwitness} (Dual witness)\\ 
For a given Tanner graph $\mathcal{T}=(V,C,E)$ and a (possibly scaled) log-likelihood ratio function $\gamma: V \to \mathbb{R}$, a dual witness $w$ is a function $w: E \to \mathbb{R}$ that satisfies the following 2 properties:
\begin{equation}\label{le:dw_var_equation}
\forall v \in V, {\displaystyle\sum\limits_{c \in N(v):w(v,c)>0} w(v,c)} < {\displaystyle\sum\limits_{c \in N(v):w(v,c) \le 0} (-w(v,c))} + \gamma(v)
\end{equation}

\begin{equation}\label{le:dw_check_equation}
\forall c \in C, \forall v, v' \in N(c), ~ w(v,c) + w(v',c) \geq 0
\end{equation}
\end{definition}

The following theorem relates the existence of a dual witness to LP decoding success. The fact that the existence of a dual witness implies LP decoding success was shown in \cite{feldman2007lp}. We prove that the converse of this statement is also true. This converse will be used in the proof of Theorem \ref{le:interplaytheorem}.
\begin{theorem}\label{le:existencedualwitness} (Existence of a dual witness and LP decoding success) \\ 
Let $\mathcal{T}=(V,C,E)$ be a Tanner graph of a binary linear code with block length $n$ and let $\eta \in \{0,1\}^{n}$ be any error pattern. Then, there is LP decoding success for $\eta$ on $\mathcal{T}$ if and only if there is a dual witness for $\eta$ on $\mathcal{T}$.
\end{theorem}

\begin{proof}[{\bf Proof of Theorem \ref{le:existencedualwitness}}]
See Appendix \ref{le:app_ex_dual_witness}.
\end{proof}

\noindent The following definition is based on Definition 1 of \cite{daskalakis2008probabilistic}.

\begin{definition}\label{le:hyperflow} (Hyperflow)\\ 
For a given Tanner graph $\mathcal{T}=(V,C,E)$ and a (possibly scaled) log-likelihood ratio function $\gamma: V \to \mathbb{R}$, a hyperflow $w$ is a function $w: E \to \mathbb{R}$ that satisfies property (\ref{le:dw_var_equation}) above as well as the following property:
\begin{equation}\label{le:hyperflow_check_equation}
\forall c \in C,  \exists P_{c} \geq 0, \exists v \in N(c) \text{ s.t. } w(v,c) = -P_{c} \text{ and } \forall v' \in N(c) \text{ s.t. } v' \neq v, w(v',c) = P_{c}
\end{equation}
\end{definition}

By Proposition $1$ of \cite{daskalakis2008probabilistic}, the existence of a hyperflow is equivalent to that of a dual witness. Hence, by Theorem \ref{le:existencedualwitness} above, we get:
 
\begin{corollary}\label{le:existencehyperflow} (Existence of a hyperflow and LP decoding success) \\ 
Let $\mathcal{T}=(V,C,E)$ be a Tanner graph of a binary linear code with block length $n$ and let $\eta \in \{0,1\}^{n}$ be any error pattern. Then, there is LP decoding success for $\eta$ on $\mathcal{T}$ if and only if there is a hyperflow for $\eta$ on $\mathcal{T}$.
\end{corollary}

\begin{definition}(WDG corresponding to a hyperflow or a dual witness)\\ 
Let $\mathcal{T}=(V,C,E)$ be a Tanner graph, $\gamma: V \to \mathbb{R}$ a (possibly scaled) log-likelihood ratio function and $w: E \to \mathbb{R}$ a dual witness or a hyperflow. The weighted directed graph (WDG) $(V,C,E,w,\gamma)$ associated with $\mathcal{T}$,$\gamma$ and $w$ has vertex set $V \cup C$ and for any $v \in V$ and any $c \in C$, an arrow is directed from $v$ to $c$ if $w(v,c) > 0$, an arrow is directed from $c$ to $v$ if $w(v,c) < 0$ and $v$ and $c$ are not connected by an arrow if $w(v,c) = 0$. Moreover, a directed edge between $v \in V$ and $c \in C$ has weight $|w(v,c)|$.
\end{definition}

The following theorem shows that whenever there exists a WDG corresponding to a hyperflow or a dual witness, there exists an acyclic WDG (denoted by WDAG) corresponding to a hyperflow.

\begin{theorem}\label{le:existenceAWDAG} (Existence of an acyclic WDG)\\ 
Let $\mathcal{T}=(V,C,E)$ be a Tanner graph of a binary linear code with block length $n$ and let $\eta \in \{0,1\}^{n}$ be any error pattern. If $G=(V,C,E,w,\gamma)$ is a WDG (Weighted Directed Graph) corresponding to a dual witness for $\eta$ on $\mathcal{T}$, then there is an acyclic WDG $G''=(V,C,E,w'',\gamma)$ corresponding to a hyperflow for $\eta$ on $\mathcal{T}$.
\end{theorem}

Before proving Theorem \ref{le:existenceAWDAG}, we summarize the different characterizations of LP decoding success.
\begin{theorem}\label{le:equivalence_summary} Let $\mathcal{T}=(V,C,E)$ be a Tanner graph of a binary linear code with block length $n$ and let $\eta \in \{0,1\}^{n}$ be any error pattern. Then, the following are equivalent:
\begin{enumerate}
\item There is LP decoding success for $\eta$ on $\mathcal{T}$.
\item There is a dual witness for $\eta$ on $\mathcal{T}$.
\item There is a hyperflow for $\eta$ on $\mathcal{T}$.
\item There is a WDAG for $\eta$ on $\mathcal{T}$.
\end{enumerate}
\end{theorem}
In order to prove Theorem \ref{le:existenceAWDAG}, we give an algorithm that transforms a WDG $G$ satisfying Equations (\ref{le:dw_var_equation}) and (\ref{le:dw_check_equation}) into an acyclic WDG $G''$ satisfying Equations (\ref{le:dw_var_equation}) and (\ref{le:hyperflow_check_equation}).

\begin{algorithm}[H]
\caption{Transforming the dual witness WDG $G$ for $\gamma$ into a hyperflow WDAG $G''$ for $\gamma$}\label{le:cycleremoval}
{\bf Input:} {$G=(V,C,E,w,\gamma)$} \\
{\bf Output:} {$G''=(V,C,E,w'',\gamma)$}\\ 
\begin{algorithmic}
\State $G'=(V,C,E,w',\gamma)$ $\gets$ $G$
\While{$G'$ has a directed cycle}
\State $c$ $\gets$ any directed cycle of $G'$
\State $w_{min}$ $\gets$ minimum weight of an edge of $c$ \Comment All edges along $c$ have a positive weight.
\State Subtract $w_{min}$ from the weights of all edges of $c$
\State Remove all zero weight edges
\State Store the resulting WDG in $G'$
\EndWhile\\ 
\For{all $j \in C$}
\State $d(j)$ $\gets$ degree of $j$
\State $\{v_{1},\dots,v_{d(j)}\}$ $\gets$ neighbours of $j$ in order of increasing $w'(v_{i},j)$
\If{$w'(v_{1},j) \geq 0$} \Comment All edges are directed toward $j$ and can thus be removed.
\State $w''(v_{i},j)$ $\gets$ $0$ $\forall i \in [d(j)]$
\Else \Comment $(v_{1},j)$ is the only edge directed away from $j$.
\State $w''(v_{1},j)$ $\gets$ $w'(v_{1},j)$
\State $w''(v_{i},j)$ $\gets$ $|w'(v_{1},j)|$ $\forall i \in \{2,\dots,d(j)\}$
\EndIf
\EndFor
\end{algorithmic}
\end{algorithm}

The next lemma is used to complete the proof of Theorem \ref{le:existenceAWDAG}.

\begin{lemma}\label{le:loopinvariant}
After each iteration of the while loop of Algorithm \ref{le:cycleremoval}, we have:
\begin{enumerate}
\item[(I)] The number of cycles of $G'$ decreases by at least $1$.
\item[(II)] $G'$ satsifies the dual witness equations (\ref{le:dw_var_equation}) and (\ref{le:dw_check_equation}).
\end{enumerate}
\end{lemma}

\begin{proof}[\bf{Proof of Lemma \ref{le:loopinvariant}}]
(I) follows from the fact that cycle $c$ is being broken in every iteration of the while loop and no new cycle is added by reducing the absolute weights of some edges of the WDG. (II) follows from the fact that during any iteration of the while loop, we are possibly repeatedly reducing the absolute weights of one ingoing and one outgoing edge of a variable or check node by the same amount, which maintains the original LP constraints (\ref{le:dw_var_equation}) and (\ref{le:dw_check_equation}). 
\end{proof}

\begin{proof}[\bf{Proof of Theorem \ref{le:existenceAWDAG}}]
First, note that the while loop of Algorithm \ref{le:cycleremoval} will be executed a number of times no larger than the number of cycles of $G$, which is finite. By Lemma \ref{le:loopinvariant}, after the last iteration of the while loop, $G'$ is an acyclic WDG that satisfies (\ref{le:dw_var_equation}) and (\ref{le:dw_check_equation}). The for loop of Algorithm \ref{le:cycleremoval} decreases the weights of edges that are directed away from variable nodes; thus, it maintains (\ref{le:dw_var_equation}) and $G''$ inherits the acyclic property of $G'$. Moreover, $G''$ satsifies (\ref{le:hyperflow_check_equation}), which completes the proof Theorem \ref{le:existenceAWDAG}.
\end{proof}

\begin{remark}
In virtue of Theorem \ref{le:existencedualwitness}, Theorem \ref{le:existencehyperflow} and Theorem \ref{le:existenceAWDAG}, we will use the terms ``hyperflow'', ``dual witness'' and ``WDAG'' interchangeably in the rest of this paper.
\end{remark}

\section[Transforming a WDAG into a directed weighted forest]{\Large{\bf Transforming a WDAG into a directed weighted forest}}\label{le:transformWDAG}

The WDAG corresponding to a hyperflow has no directed cycles but it possibly has cycles when viewed as an undirected graph. In this section, we show how to transform the WDAG corresponding to a hyperflow into a directed weighted forest (which is by definition a directed graph that is acyclic even when viewed as an undirected graph). This forest has possibly a larger number of variable and check nodes than the original WDAG but it still satisfies Equations (\ref{le:dw_var_equation}) and (\ref{le:hyperflow_check_equation}). Moreover, the vertices of the forest ``corresponding'' to a vertex of the original WDAG will have their weights sum up to the weight of the original vertex. Furthermore, the directed paths of the forest will be in a bijective correspondence with the directed paths of the original WDAG. This transformation will be used when we derive an upper bound on the weight of an edge in a WDAG of a $(d_{v},d_{c})$-regular LDPC code in Section \ref{le:maxweightregular} and of a spatially coupled code in Section \ref{le:maxweightsc_section}.

\begin{theorem} \label{le:transformation}(Transforming a WDAG into a directed weighted forest)\\ 
Let $G=(V,C,E,w,\gamma)$ be a WDAG. Then, $G$ can be transformed into a directed weighted forest $T=(V',C',E',w',\gamma')$ that has the following properties:
\begin{enumerate}
\item\label{le:rep_var_item} $V'=\underset{v \in V} \bigcup V'_{v}$ where $V'_{x} \cap V'_{y} = \emptyset$ for all $x,y \in V$ s.t. $x \neq y$. For every $v \in V$, each variable node in $V'_{v}$ is called a ``replicate'' of $v$.
\item\label{le:rep_check_item} $C'=\underset{c \in C} \bigcup C'_{c}$ where $C'_{x} \cap C'_{y} = \emptyset$ for all $x,y \in C$ s.t. $x \neq y$. For every $c \in C$, each check node in $C'_{c}$ is called a ``replicate'' of $c$.
\item\label{le:add_up_item} For all $v \in V, \displaystyle\sum\limits_{v' \in V'_{v}} \gamma'(v') = \gamma(v)$.
\item\label{le:sign_conservation} For all $v \in V$ and all $v' \in V_{v}$, $\gamma'(v')$ has the same sign as $\gamma(v)$. 
\item\label{le:sat_hyperflow_item} The forest $T$ satisfies the hyperflow equations (\ref{le:dw_var_equation}) and (\ref{le:hyperflow_check_equation}).
\item\label{le:bijection_item} The directed paths of $G$ are in a bijective correspondence with the directed paths of $T$. Moreover, if the directed path $h'$ of $T$ corresponds to the directed path $h$ of $G$, then the variable and check nodes of $h'$ are replicates of the corresponding variable and check nodes of $h$.
\item\label{le:same_weight_sink} If $G$ has a single sink node with a single incoming edge that has weight $\alpha$, then $T$ has a single sink node with a single incoming edge and that has the same weight $\alpha$.
\end{enumerate}
\end{theorem}

In order to prove Theorem \ref{le:transformation}, we now give an algorithm that transforms the WDAG $G$ into the directed weighted forest $T$.

\begin{algorithm}[H]
\caption{Transforming the WDAG $G$ into the directed weighted forest $T$}\label{le:alg_trans_tree}
{\bf Input:} {$G=(V,C,E,w,\gamma)$} \\
{\bf Output:} {$T=(V',C',E',w',\gamma')$} \\
\begin{algorithmic}
\For{$\text{each }v \in V\text{ taken in topological order}$}
\State $p$ $\gets$ $\text{number of outgoing edges of }v$
\State $\{e_j^{(v)}\}_{j=1}^{p} \gets \text{weights of outgoing edges of }v$
\State $e_{T}^{(v)} \gets \displaystyle\sum\limits_{j=1}^p e_j^{(v)}$
\State $\text{Create }p\text{ replicates of the subtree rooted at }v$ \Comment Contains all ancestors of $v$ in the current WDAG
\For{$\text{each }l \in [p]$}
\State $\text{Scale the $l$th subtree by }e_l/e_{T}^{(v)}$ \Comment The weights of all variable nodes and edges are scaled
\State $\text{Connect the $l$th subtree to the $l$th outgoing edge of $v$}$
\EndFor
\EndFor
\end{algorithmic}
\end{algorithm}

We now state and prove a loop invariant that constitutes the main part of the proof of Theorem \ref{le:transformation}. First, we introduce some notation related to the operation of Algorithm \ref{le:alg_trans_tree}.

\begin{notation}
In the following, let $V = \{v_1,\dots, v_n\}$. For every $i,j \in [n]$, let $r_{i,j}$ be the number of replicates of variable node $v_j$ after the $i$th iteration of the algorithm. Moreover, for every $k \in [r_{i,j}]$, let $v_{i,j,k}$ be the $k$th replicate of $v_j$ after the $i$th iteration of the algorithm. For all $i \in [n]$, let $V_{i}$, $C_{i}$, $E_{i}$, $\gamma_{i}$ and $w_{i}$ be the set of all variable nodes, set of all check nodes, set of all edges, log-likelihood ratio function and weight function, respectively, after the $i$th iteration of the algorithm and let $G_{i} = (V_{i},C_{i},E_{i},w_{i},\gamma_{i})$. Finally, we set $G_{0}=(V_{0},C_{0},E_{0},\gamma_{0},w_{0})$ to $(V,C,E,\gamma,w)$.
\end{notation}

\begin{lemma}\label{le:secondloopinvariant} For any $i \geq 0$, after the $i$th iteration of Algorithm \ref{le:alg_trans_tree}, we have:\footnote{By ``after the $0$th iteration'', we mean ``before the $1$st iteration''.}
\begin{enumerate}
\item[(I)] For all $j \in [n]$, $\displaystyle\sum\limits_{k=1}^{r_{i,j}} \gamma_{i}(v_{i,j,k}) = \gamma(v_{j})$.
\item[(II)] For all $j \in [n]$ and all $k \in [r_{i,j}]$, $\gamma_{i}(v_{i,j,k})$ has the same sign as $\gamma(v_{j})$.
\item[(III)] For all $v \in V_{i}$, ${\displaystyle\sum\limits_{c \in N(v):w_{i}(v,c)>0} w_{i}(v,c)} < {\displaystyle\sum\limits_{c \in N(v):w_{i}(v,c) \le 0} (-w_{i}(v,c))} + \gamma_{i}(v)$.
\item[(IV)] For all $c \in C_{i}$, there exist $P_{c} \geq 0$ and $v \in N(c)$ s.t. $w_{i}(v,c) = -P_{c}$ and for all $v' \in N(c)$ s.t. $v' \neq v, w_{i}(v',c) = P_{c}$.
\item[(V)] The directed paths of $G$ are in a bijective correspondence with the directed paths of $G_{i}$. Moreover, if the directed path $h'$ of $G_{i}$ corresponds to the directed path $h$ of $G$, then the variable and check nodes of $h'$ are replicates of the corresponding variable and check nodes of $h$.
\end{enumerate}
\end{lemma}

\begin{proof}[\bf{Proof of Lemma \ref{le:secondloopinvariant}}]
Base Case: Before the first iteration, we have: $r_{0,j} = 1\text{ ,  }\gamma_{0}(v_{0,j,1}) = \gamma(v_{j})$ for all $j \in [n]$. Thus, (I) and (II) are initially true. (III) and (IV) are initially true because the original WDAG $G$ satisfies the hyperflow equations (\ref{le:dw_var_equation}) and (\ref{le:hyperflow_check_equation}). Moreover, (V) is initially true since $G_0 = G$.\\ 
Inductive Step: We show that, for every $i \geq 1$, if (I), (III), (IV) and (V) are true after iteration $i-1$ of Algorithm \ref{le:alg_trans_tree}, then they are also true after iteration $i$.\\ 
Let $i \geq 1$. In iteration $i$, a variable node $v$ with log-likelihood ratio $\gamma_{i-1}(v)$ is (possibly) replaced by a number $p$ of replicates $\{v_{1}',\dots, v_{p}'\}$ with log-likelihood ratios $\big\{ \frac{e_l}{e_{T}^{(v)}} \gamma_{i-1}(v) ~ | ~ l \in [p] \big\}$. Therefore, the total sum of the added replicates is $ \displaystyle\sum\limits_{l=1}^p \big({ \frac{e_l}{e_{T}^{(v)}} \gamma_{i-1}(v) }\big) = \gamma_{i-1}(v)$ . Thus, (I) is true. By the induction assumption and since $e_l/e_{T}^{(v)} > 0$, it follows that (II) is also true.\\ 
To show that (III) is true, we first note that if $v' \in V_{i}$ was not created during the $i$th iteration, then $v'$ will satisfy (III) after the $i$th iteration. If $v'$ was created during the $i$th iteration, we distinguish two cases:\\ 
In the first case, $v'$ is not a replicate of $v$ (which is the variable node considered in the $i$th iteration). Then, $v'$ is a replicate of $v_{i-1} \in V_{i-1}$. By the induction assumption, $\gamma_{i-1}(v_{i-1})$ and the weights of the adjacent edges to $v_{i-1}$ satisfy (III) before the $i$th iteration. Since $\gamma_{i}(v')$ and the weights of the edges adjacent to $v'$ will be respectively equal to $\gamma_{i-1}(v_{i-1})$ and the weights of the edges adjacent to $v_{i-1}$, scaled by the same positive factor, $v'$ will satisfy (III) after the $i$th iteration.\\ 
In the second case, $v'$ is a replicate of $v$. Assume that $v'$ is the replicate of $v$ corresponding to the edge $(v,c_{0})$ where $c_{0} \in N(v)$ and $w_{i-1}(v,c_{0})>0$. During the $i$th iteration, the subtree corresponding to $v'$ will be created and in this subtree, $\gamma_{i}(v')$ and the weights of the edges incoming to $v'$ will be respectively equal to $\gamma_{i-1}(v)$ and the weights of the edges incoming to $v$, scaled by $\theta(v,c_{0}) = w_{i-1}(v,c_{0})/e_{T}^{(v)}$ where $e_{T}^{(v)} = \displaystyle\sum\limits_{c \in N(v):w_{i-1}(v,c)>0} w_{i-1}(v,c)$. The only outgoing edge of $v'$ will be $(v',c_{0})$. Thus,
\begin{align*}\displaystyle\sum\limits_{c \in N(v'):w_{i}(v',c)>0} w_{i}(v',c) = w_{i}(v',c_{0}) = w_{i-1}(v,c_{0}) &= \theta(v,c_{0}) \displaystyle\sum\limits_{c \in N(v):w_{i-1}(v,c)>0} w_{i-1}(v,c)\\ &< \theta(v,c_{0})\Big({\displaystyle\sum\limits_{c \in N(v):w_{i-1}(v,c) \le 0} (-w_{i-1}(v,c))} + \gamma_{i-1}(v)\Big)\\ & = \theta(v,c_{0}){\displaystyle\sum\limits_{c \in N(v):w_{i-1}(v,c) \le 0} (-w_{i-1}(v,c))} + \theta(v,c_{0})\gamma_{i-1}(v)\\ &= \displaystyle\sum\limits_{c \in N(v'):w_{i}(v',c) \le 0} (-w_{i}(v',c)) + \gamma_{i}(v')
\end{align*}
Therefore, $v'$ will satisfy (III) after the $i$th iteration.\\ 
Equation (IV) follows from the induction assumption and from the fact that we are either uniformly scaling the neighborhood of a check node or leaving it unchanged.\\ 
To prove that (V) is true after the $i$th iteration, let $v$ be the variable node under consideration in the $i$th iteration and consider the function that maps the directed path $h$ of $G_{i-1}$ to the directed path $h'$ of $G_{i}$ as follows:
\begin{enumerate}
\item If $h$ does not contain $v$, then $h'$ is set to $h$.
\item If $h$ contains $v$, then $h$ can be uniquely decomposed into the concatenation $h_{1} h_{2}$ where $h_{1}$ is a directed path of $G_{i-1}$ that ends at $v$ and $h_{2}$ is a directed path of $G_{i-1}$ that starts at $v$. Let $e_{l}$ be the first edge of $h_{2}$. Then, $h'$ is set to $h'_{1} h_{2}$ where $h'_{1}$ is the directed path in the $l$th created subtree of $G'$ that corresponds to $h_{1}$.
\end{enumerate}
This map is a bijection from the set of all directed paths of $G_{i-1}$ to the set of all directed paths of $G_{i}$. Moreover, if the directed path $h$ of $G_{i-1}$ is mapped to the directed path $h'$ of $G_{i}$, then the variable and check nodes of $h'$ are replicates of the corresponding variable and check nodes of $h$.

\end{proof} 

\begin{proof}[\bf{Proof of Theorem \ref{le:transformation}}]
Note that \ref{le:rep_var_item} and \ref{le:rep_check_item} in Theorem \ref{le:transformation} follow from the operation of Algorithm \ref{le:alg_trans_tree}. Moreover, \ref{le:add_up_item}, \ref{le:sign_conservation}, \ref{le:sat_hyperflow_item} and \ref{le:bijection_item} follow from Lemma \ref{le:secondloopinvariant} with $\gamma' = \gamma_{n}$. To prove \ref{le:same_weight_sink}, note that if $G$ has a single sink node $v$, then $v$ will be the last vertex in any topological ordering of the vertices of $G$. Furthermore, if $v$ has a single incoming edge with weight $\alpha$, then it will have only one replicate in $T$, with a single incoming edge having the same weight $\alpha$.\\ 
\end{proof}

\section[Maximum weight of an edge in a regular WDAG on the BSC]{\Large{\bf Maximum weight of an edge in a regular WDAG on the BSC}}\label{le:maxweightregular}

In this section, we present sublinear (in the block length $n$) upper bound on the weight of an edge in a regular WDAG. The main idea of the proof is the following. Consider a $(d_{v},d_{c})$-regular WDAG $G$ (where $d_{v},d_{c}>2$ are constants) corresponding to a hyperflow. Note that each variable node has a log-likelihood ratio of $\pm 1$. Thus, the total amount of flow available in the WDAG is most $n$. Moreover, for a substantial weight to get ``concentrated'' on an edge in the WDAG, the $+1$'s should ``move'' from variable nodes accross the WDAG toward that edge. By the hyperflow equation (\ref{le:hyperflow_check_equation}), each check node cuts its incoming flow by a factor of $d_{c}-1$. Thus, it can be seen that the maximum weight that can get concentrated on an edge is asymptotically smaller than $n$.

\begin{theorem}\label{le:maxweight} (Maximum weight of an edge in a regular WDAG on the BSC)\\ 
Let $G=(V,C,E,w,\gamma)$ be a WDAG corresponding to LP decoding of a $(d_{v},d_{c})$-regular LDPC code (with $d_{v},d_{c}>2$) on the BSC. Let $n=|V|$ and $\alpha_{max}=\underset{e \in E}{\operatorname{max}}{ |w(e)| }$ be the maximum weight of an edge in $G$. Then,
\begin{equation} 
\alpha_{max} \le c n^{\frac{\ln(d_{v}-1)}{\ln(d_{v}-1)+\ln(d_{c}-1)}} = o(n)
\end{equation}
for some constant $c > 0$ depending only on $d_{v}$.
\end{theorem}

We now state and prove a series of lemmas that leads to the proof of Theorem \ref{le:maxweight}.

\begin{definition}\label{le:reverseorientedtree} (Root-oriented tree)\\ 
A root-oriented tree is defined in the same way as the WDAG in Definition \ref{le:hyperflow} and Theorem \ref{le:existenceAWDAG} but with the further constraints that $T$ has a single sink node (which is a variable node) and that $T$ is a tree when viewed as an undirected graph. Note that the name ``root-oriented'' is due to the fact that the edges are oriented toward the root of the tree, as shown in Figure \ref{le:root_oriented_tree_figure}.
\end{definition}

\begin{remark}
Algorithm \ref{le:alg_trans_tree} can also be used to generate the directed weighted forest corresponding to the subset of the WDAG consisting of all variable and check nodes that are ancestors of a given variable node $v$. In this case, the output is a root-oriented tree with its single sink node being the unique replicate of $v$.
\end{remark}

\begin{figure}[!h]

\centering

\begin{tikzpicture}
  [scale=.8,auto=left,every node/.style={circle,fill=gray!20}]
  \node (n1) at (12,10) {$v_{0}$};
  \node (n2) at (10,9)  {$c_{1}$};
  \node (n3) at (14,9)  {$c_{2}$};
  \node (n4) at (9,8) {$v_{1}$};
  \node (n5) at (11,8)  {$v_{2}$};
  \node (n6) at (13,8)  {$v_{3}$};
  \node (n7) at (15,8)  {$v_{4}$};

  \foreach \from/\to in {n2/n1,n3/n1,n4/n2,n2/n1,n5/n2,n6/n3,n7/n3}
    \draw[->] (\from) -- (\to);

\end{tikzpicture}

\caption{Root-oriented tree with root the variable node $v_{0}$}\label{le:root_oriented_tree_figure}

\end{figure}
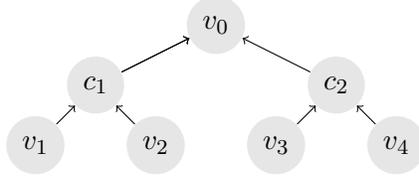

\begin{definition}\label{le:Gmax} ($G_{max}$, $\alpha_{max}$)\\ 
Let $G=(V,C,E,w,\gamma)$ be a WDAG. Let $e_{max}=(v_{max},c_{max}) = \underset{(v,c):w(v,c) \le 0}{\operatorname{argmax}}{~ |w(v,c)|}$ and let $\alpha_{max} = |w(v_{max},c_{max})|$. Let $V_{max} = V_{1} \cup \{v_{max}\}$ where $V_{1}$ is the set of all variable nodes $v \in V$ s.t. $c_{max}$ is reachable from $v$ in $G$ and let $C_{max}$ be the set of all check nodes $c \in C$ s.t. $c_{max}$ is reachable from $c$ in $G$.\footnote{Note that $c_{max} \in C_{max}$.} Let $G_{max}=(V_{max},C_{max},E_{max},w_{max},\gamma_{max})$ be the corresponding WDAG.
\end{definition}

\begin{definition}\label{le:depthreverseoriented} (Depth of a variable node in a root-oriented tree)\\ 
Let $T$ be a root-oriented tree with root $v_0$. For any variable node $v$ in $T$, the depth of $v$ in $T$ is defined to be the number of check nodes on the unique directed path from $v$ to $v_0$ in $T$.
\end{definition}

\begin{definition}\label{le:phifunction} ($F$-function)\\ 
Let $G=(V,C,E,w,\gamma)$ be a WDAG. For any $S \subseteq V$, define $F(S) = \displaystyle\sum\limits_{v \in S} \displaystyle\sum\limits_{c \in N(v):w(v,c) \geq 0} w(v,c)$. In other words, $F(S)$ is the sum of all the ``flow'' leaving variable nodes in $S$ to adjacent check nodes.
\end{definition}

\begin{lemma}\label{le:firstinductivelemma}
Let $G=(V,C,E,w,\gamma)$ be a WDAG corresponding to LP decoding of a $(d_{v},d_{c})$-regular LDPC code (with $d_{v},d_{c}>2$) on the BSC and let $G_{max}=(V_{max},C_{max},E_{max},w_{max},\gamma_{max})$ be the WDAG corresponding to Definition \ref{le:Gmax}. Let $n_{max}=|V_{max}|$ and $T=(V',C',E',w',\gamma')$ be the output of Algorithm \ref{le:alg_trans_tree} on input $G_{max}$. Note that $T$ is a root-oriented tree with root $v_{max}$ which has a single incoming edge with weight $\alpha_{max}$ (by Theorem \ref{le:transformation}). Let $d_{max}$ be the maximum depth of a variable node in $T$ and for any $m \in \{0,\dots,d_{max}\}$, let $S_{m}$ be the set of all variable nodes in $T$ with depth equal to $m$. Moreover, for all $i \in \{0,\dots,d_{max}\}$ and all $j \in [n_{max}]$, let $d_{i,j}$ denote the number of replicates of variable node $v_j$ having depth equal to $i$ in $T$. Furthermore, for every $k \in [d_{i,j}]$, let $\Gamma_{i,j,k}$ be the $\gamma'$ value of the $k$th replicate of $v_j$ among those having depth equal to $i$ in $T$. Then, for all $m \in \{1,\dots,d_{max}\}$, we have:
\begin{equation}\label{le:inductive_proposition}{(P_{m}): ~ ~ ~ F(S_{m}) \geq (d_{c}-1)^{m} \alpha_{max} - \displaystyle\sum\limits_{i=0}^{m-1} (d_{c}-1)^{m-i}\displaystyle\sum\limits_{j=1}^{n_{max}} \displaystyle\sum\limits_{k=1}^{d_{i,j}}  \Gamma_{i,j,k}}\end{equation}
\end{lemma}

\begin{proof}[\bf{Proof of Lemma \ref{le:firstinductivelemma}}]
For any $S \subseteq V'$, let $\Delta(S)$ be the set of all $v \in V'$ for which there exist $s \in S$ and a directed path from $v$ to $s$ in $T$ containing exactly one check node. We proceed by induction on $m$.\\ 
Base Case: $m=1$. We note that $S_{1}=\Delta(\{v_{max}\})$ and that $v_{max}$ is the only variable node in $T$ having depth equal to $0$ in $T$. Hence, for the hyperflow to satisfy (\ref{le:hyperflow_check_equation}), we should have:
\begin{equation*}F(S_{1}) \geq (d_{c}-1)(\alpha_{max}-\gamma'(v_{max})) = (d_{c}-1)\alpha_{max} - \displaystyle\sum\limits_{i=0}^0 (d_{c}-1)^{1}\displaystyle\sum\limits_{j=1}^{n_{max}} \displaystyle\sum\limits_{k=1}^{d_{i,j}}  \Gamma_{i,j,k} \end{equation*}
Note that the last equality follows from the facts that $d_{0,j} = 1$ if $v_{j} = v_{max}$ and $d_{0,j} = 0$ otherwise, and that $\Gamma_{i,j,k} = \gamma'(v_{max})$ if $v_{j} = v_{max}$ and $k=1$ and $\Gamma_{i,j,k} = 0$ otherwise.\\ 
Inductive Step: We need to show that if $(P_{m})$ is true for some $1 \le m \le d_{max}-1$, then $(P_{m+1})$ is also true. Assuming that $(P_{m})$ is true, $S_{m}$ satisfies Equation (\ref{le:inductive_proposition}). Since $T$ is a root-oriented tree, $S_{m+1}=\Delta(S_{m})$. Hence, for the hyperflow to satisfy (\ref{le:hyperflow_check_equation}), we should have:
\begin{align*}
F(S_{m+1}) &\geq (d_{c}-1) \big(F(S_{m}) - \displaystyle\sum\limits_{j=1}^{n_{max}} \displaystyle\sum\limits_{k=1}^{d_{m,j}} \Gamma_{m,j,k} \big)\\ 
&\geq (d_{c}-1)[(d_{c}-1)^{m} \alpha_{max} - \displaystyle\sum\limits_{i=0}^{m-1} (d_{c}-1)^{m-i}\displaystyle\sum\limits_{j=1}^{n_{max}} \displaystyle\sum\limits_{k=1}^{d_{i,j}}  \Gamma_{i,j,k} - \displaystyle\sum\limits_{j=1}^{n_{max}} \displaystyle\sum\limits_{k=1}^{d_{m,j}} \Gamma_{m,j,k} ]\\ 
&= (d_{c}-1)^{m+1} \alpha_{max} - \displaystyle\sum\limits_{i=0}^m (d_{c}-1)^{m+1-i}\displaystyle\sum\limits_{j=1}^{n_{max}} \displaystyle\sum\limits_{k=1}^{d_{i,j}} \Gamma_{i,j,k}
\end{align*}
\end{proof}

\begin{definition}\label{le:depthsingle} (Depth of a variable node in a WDAG with a single sink node)\\ 
Let $G=(V,C,E,w,\gamma)$ be a WDAG with a single sink node $v_0 \in V$ and let $v \in V$. The depth of $v$ in $G$ is defined to be the minimal number of check nodes on a directed path from $v$ to $v_0$ in $G$.
\end{definition}

\begin{corollary}\label{le:firstmaximization}
Let \ensuremath{g_{max}} be the maximum depth of a variable node $v \in V_{max}$ in the WDAG $G_{max}$ (which has a single sink node $v_{max}$).\footnote{Note that in general $g_{max} \le d_{max}$ but the two quantities need not be equal.} Then,
\begin{equation}\label{le:max_equation} \alpha_{max} \le \max_{ (T_{0}, \dots , T_{g_{max}}) \in W} { f(T_{0}, \dots , T_{g_{max}}) } \end{equation}
where: 
\begin{equation*} f(T_{0}, \dots , T_{g_{max}}) = \displaystyle\sum\limits_{i=0}^{g_{max}} \frac{T_{i}}{(d_{c}-1)^{i}} \end{equation*}
and $W$ is the set of all tuples $(T_{0}, \dots , T_{g_{max}}) \in \mathbb{N}^{g_{max}+1}$ satisfying the following three equations:
\begin{equation}\label{le:sum_equation}
\displaystyle\sum\limits_{i=0}^{g_{max}} T_{i} = n_{max}
\end{equation}
\begin{equation}\label{le:base_equation}
T_{0} = 1
\end{equation}
\begin{equation}\label{le:inductive_equation}
\text{For all } i \in \{0,\dots,g_{max}-1\}, ~ T_{i+1} \le (d_{c}-1)(d_{v}-1)T_{i}
\end{equation}
\end{corollary}

\begin{proof}[\bf{Proof of Corollary \ref{le:firstmaximization}}]
Setting $m=d_{max}$ in Lemma \ref{le:firstinductivelemma} and noting that the leaves of $T$ have no entering flow, we get:
\begin{equation*}
\displaystyle\sum\limits_{j=1}^{n_{max}} \displaystyle\sum\limits_{k=1}^{d_{d_{max},j}}  \Gamma_{d_{max},j,k} \geq F(S_{d_{max}}) \geq (d_{c}-1)^{d_{max}} \alpha_{max} - \displaystyle\sum\limits_{i=0}^{d_{max}-1} (d_{c}-1)^{d_{max}-i}\displaystyle\sum\limits_{j=1}^{n_{max}} \displaystyle\sum\limits_{k=1}^{d_{i,j}} \Gamma_{i,j,k}
\end{equation*}
Thus,
\begin{equation*}
\alpha_{max} \le \displaystyle\sum\limits_{i=0}^{d_{max}} \frac{1}{(d_{c}-1)^{i}}\displaystyle\sum\limits_{j=1}^{n_{max}} \displaystyle\sum\limits_{k=1}^{d_{i,j}}  \Gamma_{i,j,k}
\end{equation*}
Part \ref{le:bijection_item} of Theorem \ref{le:transformation} implies that for all $v \in V_{max}$, the depth of $v$ in $G_{max}$ is equal to the minimum depth in $T$ of a replicate of $v$. By parts \ref{le:add_up_item} and \ref{le:sign_conservation} of Theorem \ref{le:transformation}, we also have that for all $j \in [n_{max}]$, $\displaystyle\sum\limits_{i=0}^{d_{max}} \displaystyle\sum\limits_{k=1}^{d_{i,j}} \Gamma_{i,j,k} \le 1$ and for all $i \in \{0,\dots,d_{max}\}$ and all $k \in [d_{i,j}]$, $\Gamma_{i,j,k} \le 1$ and $\{\Gamma_{i,j,k}\}_{i,k}$ all have the same sign. For every $ j\in [n_{max}]$, let $d_{j}$ be the depth of $v_j$ in $G_{max}$ and note that $d_j \le i$ for every $i \in \{0,\dots,d_{max}\}$ for which there exists $k \in [d_{i,j}]$ s.t. $\Gamma_{i,j,k} \neq 0$. Thus, we get that:
\begin{equation*}
\alpha_{max} \le \displaystyle\sum\limits_{i=0}^{d_{max}} \frac{1}{(d_{c}-1)^{i}}\displaystyle\sum\limits_{j=1}^{n_{max}} \displaystyle\sum\limits_{k=1}^{d_{i,j}}  |\Gamma_{i,j,k}| \le \displaystyle\sum\limits_{j=1}^{n_{max}} \frac{1}{(d_{c}-1)^{d_{j}}} \displaystyle\sum\limits_{i=0}^{d_{max}} \displaystyle\sum\limits_{k=1}^{d_{i,j}}  |\Gamma_{i,j,k}| = \displaystyle\sum\limits_{i=0}^{d_{max}} \frac{1}{(d_{c}-1)^{i}} T_i
\end{equation*}
where the last equality follows from the fact that $\displaystyle\sum\limits_{i=0}^{d_{max}} \displaystyle\sum\limits_{k=1}^{d_{i,j}}  |\Gamma_{i,j,k}| = |\displaystyle\sum\limits_{i=0}^{d_{max}} \displaystyle\sum\limits_{k=1}^{d_{i,j}} \Gamma_{i,j,k} | = 1$ for every $j \in [n_{max}]$ with $T_i$ being the number of variable nodes with depth equal to $i$ in $G_{max}$ for every $i \in [d_{max}]$. Note that the notion of depth used here is the one given in Definition \ref{le:depthsingle} since $G_{max}$ is a WDAG with a single sink node $v_{max}$. Since $T_{i} =0$ for all $g_{max} < i \le d_{max}$, we get:
\begin{equation*}
\alpha_{max} \le \displaystyle\sum\limits_{i=0}^{g_{max}} \frac{1}{(d_{c}-1)^{i}} T_i
\end{equation*}
Equations (\ref{le:sum_equation}), (\ref{le:base_equation}) and (\ref{le:inductive_equation}) follow from the definitions of $T_i$ and $g_{max}$.
\end{proof}

\begin{lemma}\label{le:finalupperbound}
The RHS of Equation (\ref{le:max_equation}) is at most $c \times (n_{max})^{\frac{\ln(d_{v}-1)}{\ln(d_{v}-1)+\ln(d_{c}-1)}}$ for some constant $c > 0$ depending only on $d_{v}$.
\end{lemma}

\begin{proof}[\bf{Proof of Lemma \ref{le:finalupperbound}}]
Follows from Theorem \ref{le:unified_theorem} with $\lambda = 1$, $\beta = (d_{c}-1)(d_{v}-1)$ and $m=n_{max}$.
\end{proof}

\begin{proof}[\bf{Proof of Theorem \ref{le:maxweight}}]
Theorem \ref{le:maxweight} follows from Corollary \ref{le:firstmaximization} and Lemma \ref{le:finalupperbound} by noting that $|V_{max}| \le |V|$ since $V_{max} \subseteq V$ and that $\underset{e \in E}{\operatorname{max}}{~ |w(e)|} = \Omega(\underset{(v,c):w(v,c) \le 0}{\operatorname{max}}{~ |w(v,c)|})$ by the hyperflow equation (\ref{le:hyperflow_check_equation}).
\end{proof}

We now show that the bound given in Theorem \ref{le:maxweight} is asymptotically tight in the case of $(d_{v},d_{c})$-regular LDPC codes.
\begin{theorem}\label{le:asymptotictightness} (Asymptotic tightness of Theorem \ref{le:maxweight} for $(d_{v},d_{c})$-regular LDPC codes)\\ 
There exists an infinite family of $(d_{v},d_{c})$-regular Tanner graphs $\{(V_{n},C_{n},E_{n})\}_n$, an infinite family of error patterns $\{\gamma_{n}\}_n$ and a positive constant $c$ s.t. there exists a hyperflow for $\gamma_{n}$ on $(V_{n},C_{n},E_{n})$ and any WDAG $(V_{n},C_{n},E_{n},w,\gamma_{n})$ corresponding to a hyperflow for $\gamma_{n}$ on $(V_{n},C_{n},E_{n})$ must have
\begin{equation*}
\underset{e \in E_{n}}{\operatorname{max}}{ |w(e)| } \geq c n^{\frac{\ln(d_{v}-1)}{\ln(d_{v}-1)+\ln(d_{c}-1)}}
\end{equation*}
\end{theorem}

\begin{proof}[\bf{Proof of Theorem \ref{le:asymptotictightness}}]
See Appendix \ref{le:app_asymptotic_tightness}.
\end{proof}

\section[Maximum weight of an edge in the WDAG of a spatially coupled code on the BSC]{\Large{\bf Maximum weight of an edge in the WDAG of a spatially coupled code on the BSC}}\label{le:maxweightsc_section}

The upper bound of Theorem \ref{le:maxweight} holds for $(d_{v},d_{c})$-regular LDPC codes. In this section, we derive a similar sublinear (in the block length $n$) upper bound that holds for spatially coupled codes. 

\begin{theorem}\label{le:maxweightsc} (Maximum weight of an edge in a spatially coupled code)\\ 
Let $G=(V,C,E,w,\gamma)$ be a WDAG corresponding to LP decoding of any code of the $(d_{v},d_{c}=kd_{v},L,M)$ spatially coupled ensemble on the BSC. Let $n=(2L+1)M=|V|$ be the block length of the code. Let $\alpha_{max}=\underset{e \in E}{\operatorname{max}}{ |w(e)| }$ be the maximum weight of an edge in $G$. Then,
\begin{equation} 
\ensuremath{\alpha_{max} \le c n^{\frac{\ln(q)-\ln(d_{c}-1)}{\ln(q)}} = c n^{1-\epsilon}} = o(n)
\end{equation}
for some constant $c > 0$ depending only on $d_{v}$ and where $q = d_{v}(d_{c}-1)\frac{(d_{v}-1)^{d_{v}}-1}{d_{v}-2}$ and $0< \epsilon = \frac{\ln(d_{c}-1)}{\ln(q)} <1$.
\end{theorem}

We now state and prove a series of lemmas that leads to the proof of Theorem \ref{le:maxweightsc}. Note that a central idea in the proof of Section \ref{le:maxweightregular} is that all check nodes being $d_{c}$-regular in that case, the flow at every check node is ``cut'' by a factor of $d_{c}-1$. On the other hand, a $(d_{v}=3,d_{c}=6,L,M)$ spatially coupled code has $2M$ check nodes with degree $2$ and the flow is preserved at such check nodes. To show that even in this case, the maximum weight of an edge is sublinear in the block length, we argue that a check node that is not $d_{c}$-regular should have a $d_{c}$-regular check node that is ``close by'' in the WDAG. To simplify the argument, we first ``clean'' the WDAG of the spatially coupled code to obtain a ``reduced WDAG'' with all check nodes having either degree $d_{c}$ or degree $2$. We also use a notion of ``regular check depth'' which is the same as the notion of depth of Section \ref{le:maxweightsc} except that only $d_{c}$-regular check nodes are now counted.
\begin{definition}\label{le:reducedWDAG} (Reduced WDAG)\\ 
Let $G=(V,C,E,w,\gamma)$ be a WDAG and $G_{max}=(V_{max},C_{max},E_{max},w_{max},\gamma_{max})$ be the WDAG corresponding to Definition \ref{le:Gmax}. The reduced WDAG $G_{r}$ of $G_{max}$ is obtained by processing $G_{max}$ as follows so that each check node has either degree $d_{c}$ or degree $2$:
\begin{enumerate}
\item\label{le:left_cleaning} For every check node $c$ of $G_{r}$ with spatial index\footnote{The notion of ``spatial index'' used here is the one from Definition \ref{le:spatiallycoupledensemble}.} $<(-L+\hat{d_{v}})$, we remove all the incoming edges to $c$ except one that comes from a parent\footnote{The notion of ``parent'' of a node is the one induced by the direction of the edges of $G_r$.} of $c$ having maximal spatial index.
\item\label{le:right_cleaning} For every check node $c$ of $T'$ with spatial index $>(L-\hat{d_{v}})$, we remove all the incoming edges to $c$ except for one edge that comes from a parent of $c$ having minimal spatial index.
\item We keep only the variable nodes $v$ s.t. $v_{max}$ is still reachable from $v$ and the check nodes $c$ s.t. $v_{max}$ is still reachable from $c$.
\end{enumerate}
Note that in steps \ref{le:left_cleaning} and \ref{le:right_cleaning} above, the check nodes of $G_{r}$ are considered in an arbitrary order.
\end{definition}

\begin{definition}\label{le:reducedtree} (Reduced tree)\\ 
A reduced tree with root $v_{0}$ is a root-oriented tree with root $v_{0}$ and where every check node has either degree $d_{c}$ or degree $2$.
\end{definition}
\noindent Note that if we run Algorithm \ref{le:alg_trans_tree} on a reduced WDAG, the output will be a reduced tree.

\begin{definition}\label{le:regularcheckdepthvar}(Regular check depth of a variable node in a reduced tree)\\  
Let $T$ be a reduced tree with root $v_{0}$. For any variable node $v$ of $T$, the regular check depth of $v$ in $T$ is the number of $d_{c}$-regular check nodes on the directed path from $v$ to $v_{0}$ in $T$.
\end{definition}

\begin{lemma}\label{le:inductivelemma2}
Let $G=(V,C,E,w,\gamma)$ be a WDAG corresponding to LP decoding of a spatially coupled code on the BSC, $G_{max}=(V_{max},C_{max},E_{max},w_{max},\gamma_{max})$ be the WDAG corresponding to Definition \ref{le:Gmax}, $G_{r}=(V_{r},C_{r},E_{r},w_{r},\gamma_{r})$ be the reduced WDAG corresponding to $G_{max}$ and $T=(V'_{r},C'_{r},E'_{r},w'_{r},\gamma'_{r})$ be the output of Algorithm \ref{le:alg_trans_tree} on input $G_{r}$. Let $n_r=|V_{r}|$. Note that $T$ is a reduced tree with root $v_{max}$ which has a single incoming edge with weight $\alpha_{max}$ (by Theorem \ref{le:transformation}). Let $r_{max}$ be the maximum regular check depth in $T$ of a variable node $v \in V'_{r}$. For all $i \in \{0,\dots,r_{max}\}$ and all $j \in [n_r]$, let $y_{i,j}$ be the number of replicates of variable node $v_j$ having regular check depth equal to $i$ in $T$. Moreover, for all $k \in [y_{i,j}]$, let $\Gamma_{i,j,k}$ denote the $\gamma'_{r}$ value of the $k$th replicate of $v_j$ among those having regular check depth equal to $i$ in $T$. Then, for all $m \in \{1,\dots,r_{max}\}$, we have:\\ \\ 
$(P_{m})$: There exists \begin{math} U_{m} \subseteq V'_{r} \end{math} consisting of variable nodes having regular check depth $m$ in $T$ and s.t. all variable nodes of $T$ having regular check depth between $m+1$ and $r_{max}$ (inclusive) are ancestors of $U_{m}$ in $T$ and s.t.:
\begin{equation}\label{le:sec_ineq}{F(U_{m}) \geq (d_{c}-1)^{m} \alpha_{max} - \displaystyle\sum\limits_{i=0}^{m-1} (d_{c}-1)^{m-i}\displaystyle\sum\limits_{j=1}^{n_r} \displaystyle\sum\limits_{k=1}^{y_{i,j}}  \Gamma_{i,j,k} }\end{equation}
\end{lemma}

\begin{proof}[\bf{Proof of Lemma \ref{le:inductivelemma2}}]
For any $S \subseteq V'_{r}$, let $\Delta(S)$ be the set of all $v \in V'_{r}$ for which there exist $s \in S$ and a directed path from $v$ to $s$ in $T$ with the child of $v$ on this path being the unique $d_{c}$-regular check node on the path.\footnote{Again, the notion of ``child'' here is the one induced by the direction of the edges of $T$.} We proceed by induction on $m$.\\ 
Base Case: {$m=1$}. Let $U_{1}=\Delta(\{v_{max}\})$. Note that the ancestors of $v_{max}$ (inlcuding $v_{max}$) that are proper descendants of nodes in $U_{1}$ are exactly those variable nodes having regular check depth equal to $0$ in $T$. Hence, for the hyperflow to satisfy Equation (\ref{le:hyperflow_check_equation}), we should have:
\begin{equation*}
F(U_{1}) \geq (d_{c}-1) \big( \alpha_{max} - \displaystyle\sum\limits_{j=1}^{n_r} \displaystyle\sum\limits_{k=1}^{y_{0,j}}  \Gamma_{0,j,k} \big) = (d_{c}-1)^{1}\alpha_{max} - \displaystyle\sum\limits_{i=0}^0 (d_{c}-1)^{1}\displaystyle\sum\limits_{j=1}^{n_r} \displaystyle\sum\limits_{k=1}^{y_{i,j}}  \Gamma_{i,j,k}
\end{equation*}
Inductive Step: We need to show that if $(P_{m})$ is true for some $1 \le m \le (r_{max}-1)$ then $(P_{m+1})$ is also true. Assuming that $(P_{m})$ is true, there exists \begin{math}U_{m} \subseteq V'_{r} \end{math} that satisfies Equation (\ref{le:sec_ineq}) and s.t. $U_{m}$ consists of variable nodes having regular check depth $m$ in $T$, and all variable nodes of $T$ with regular check depth between $m+1$ and $r_{max}$ (inclusive) are ancestors of $U_{m}$ in $T$. Let $U_{m+1}=\Delta(U_{m})$. Note that the variable nodes that are ancestors of nodes in $U_{m}$ and proper descendants of nodes in $U_{m+1}$ are exactly those having regular check depth equal to $m$ in $T$. Hence, for the hyperflow to satisfy Equation (\ref{le:hyperflow_check_equation}), we should have:
\begin{align*}
F(U_{m+1}) &\geq (d_{c}-1) \big( F(U_{m}) - \displaystyle\sum\limits_{j=1}^{n_r} \displaystyle\sum\limits_{k=1}^{y_{m,j}}  \Gamma_{m,j,k} \big)\\ 
&\geq (d_{c}-1)[(d_{c}-1)^{m} \alpha_{max} - \displaystyle\sum\limits_{i=0}^{m-1} (d_{c}-1)^{m-i}\displaystyle\sum\limits_{j=1}^{n_r} \displaystyle\sum\limits_{k=1}^{y_{i,j}}  \Gamma_{i,j,k} - \displaystyle\sum\limits_{j=1}^{n_r} \displaystyle\sum\limits_{k=1}^{y_{m,j}}  \Gamma_{m,j,k}]\\ 
&= (d_{c}-1)^{m+1} \alpha_{max} - \displaystyle\sum\limits_{i=0}^{m-1} (d_{c}-1)^{m+1-i}\displaystyle\sum\limits_{j=1}^{n_r} \displaystyle\sum\limits_{k=1}^{y_{i,j}}  \Gamma_{i,j,k} - (d_{c}-1)\displaystyle\sum\limits_{j=1}^{n_r} \displaystyle\sum\limits_{k=1}^{y_{m,j}} \Gamma_{m,j,k}\\ 
&= (d_{c}-1)^{m+1} \alpha_{max} - \displaystyle\sum\limits_{i=0}^{m} (d_{c}-1)^{m+1-i}\displaystyle\sum\limits_{j=1}^{n_r} \displaystyle\sum\limits_{k=1}^{y_{i,j}}  \Gamma_{i,j,k}
\end{align*}
\end{proof}

\begin{definition}\label{le:regularcheckdepthsingle} (Regular check depth of a variable node in a reduced WDAG)\\ 
Let $G_{r}$ be a reduced WDAG with its single sink node denoted by $v_{0}$. For any variable node $v$ of $G_{r}$, the regular check depth of $v$ in $G_{r}$ is the minimum number of $d_{c}$-regular check nodes on a directed path from $v$ to $v_{0}$ in $G_{r}$.
\end{definition}

\begin{lemma}\label{le:q_expression}
Let $G_r$ be a reduced WDAG and $z_{max}$ be the maximum regular check depth of a variable node in $G_{r}$. For all $i \in \{0,\dots,z_{max}\}$, let $T_{i}$ be the number of variable nodes in $G_{r}$ with regular check depth equal to $i$. Then, for all $i \in \{0,\dots,z_{max}-1\}$:
\begin{equation*}
T_{i+1} \le qT_{i}
\end{equation*}
where $q = d_{v}(d_{c}-1)\frac{(d_{v}-1)^{d_{v}}-1}{d_{v}-2}$. Moreover, $T_{0} \le 1 + \frac{(d_{v}-1)^{d_{v}-1}-1}{d_{v}-2} = q_{0}$.
\end{lemma}

\begin{proof}[\bf{Proof of Lemma \ref{le:q_expression}}]
If, for any $i \in \{0,\dots,z_{max}\}$, we let $W_{i}$ be the set of all variable nodes in $G_{r}$ with regular check depth equal to $i$, then $T_{i} = |W_{i}|$. Fix $i \in \{0,\dots,z_{max}-1\}$. For a variable node $v$ of $G_{r}$, define $\Delta'(v)$ to be the set of all variable nodes $v_{0}$ in $G_{r}$ s.t. there exists a directed path $\mathcal{P}$ from $v_{0}$ to $v$ in $G_{r}$ s.t. the parent of $v$ on $\mathcal{P}$ is the only $d_{c}$-regular check node on $\mathcal{P}$. Note that for every variable node $u \in W_{i+1}$, there exists a variable node $v \in W_{i}$ s.t. $u \in \Delta'(v)$. Thus, $W_{i+1} \subseteq \underset{v \in W_{i}} \bigcup \Delta'(v)$ which implies that
\begin{equation*}
|W_{i+1}| \le |W_{i}| \times \underset{v \in W_{i}} \max |\Delta'(v)| \le |W_{i}| \times \underset{v \in V_r} \max |\Delta'(v)|
\end{equation*}
where $V_r$ is the set of all variable nodes of $G_{r}$. We now show that for every $v \in V_r$, $|\Delta'(v)| \le q$. Fix $v \in V_{r}$. We claim that for all $u \in \Delta'(v)$, there exists a directed path from $u$ to $v$ in $G_{r}$ containing a single $d_{c}$-regular check node which is the parent of $v$ on this path and at most $(d_{v}-1)$ $2$-regular check nodes. To show this, let $\mathcal{P}$ be a directed path from $u$ to $v$ in $G_{r}$ containing no $d_{c}$-regular check nodes other than the parent of $v$ on this path. If $\mathcal{P}$ does not contain any $2$-regular check nodes, then the needed property holds. If $\mathcal{P}$ contains at least one $2$-regular check node, then,
\begin{equation}\label{le:path_equation}
\mathcal{P}: u \leadsto c_1 \leadsto v_1 \leadsto c_2 \leadsto v_2 \leadsto \dots \leadsto c_l \leadsto v_l \leadsto c_{*} \leadsto v
\end{equation}
where $l$ is a positive integer, $c_1, c_2,\dots,c_l$ are $2$-regular check nodes of $G_r$, $c_{*}$ is a $d_{c}$-regular check node of $G_{r}$ and $v_1,v_2,\dots,v_{l}$ are variable nodes of $G_r$. For any check node $c$, we denote by $si(c)$ the spatial index of $c$. Since $c_1$ is $2$-regular, its spatial index $si(c_{1})$ is either in the interval $[-L-\hat{d_{v}}:-L+\hat{d_{v}}-1]$ or in the interval $[L-\hat{d_{v}}+1:L+\hat{d_{v}}]$. Without loss of generality, assume that $si(c_1) \in [L-\hat{d_{v}}+1:L+\hat{d_{v}}]$. For any $i \in \{0,\dots,l-1\}$, Definition \ref{le:reducedWDAG} implies that $v_i$ is at a minimal position w.r.t. $c_{i+1}$. By Definition \ref{le:spatiallycoupledensemble}, if variable node $v$ is at a minimal position w.r.t. check node $c$, then $c$ is at a maximal position w.r.t. $v$. So for any $i \in \{0,\dots,l-1\}$, $c_{i+1}$ is at a maximal position w.r.t $v_i$ and thus $si(c_i) \le si(c_{i+1})$. By condition \ref{le:additional_constraint} of Definition \ref{le:spatiallycoupledensemble}, variable node $v_i$ is not connected to two check nodes at the same position, which implies that $si(c_i) \neq si(c_{i+1})$ for all $i \in \{0,\dots,l-1\}$. So we conclude that $si(c_i) < si(c_{i+1})$ for all $i \in \{0,\dots,l-1\}$. Therefore,
\begin{equation*}
L-\hat{d_{v}}+1 \le si(c_1) < si(c_2) < \dots < si(c_l) \le L+\hat{d_{v}}
\end{equation*}
Hence, $l \le 2\hat{d_{v}} = d_{v}-1$. So $\mathcal{P}$ satisfies the needed property.\\ 
For all $i \in [d_{v}-1]$, let $n_i$ be the number of variable nodes $u$ in $G_r$ for which the smallest integer $l$ for which Equation (\ref{le:path_equation}) holds is $l=i$. Also, let $n_0$ be the number of variable nodes $u$ in $G_r$ for which there exists a path $\mathcal{P}$ of the form 
\begin{equation}\label{le:no_2_regular_path}
\mathcal{P}: u \leadsto c_{*} \leadsto v
\end{equation}
where $c_{*}$ is a $d_{c}$-regular check node of $G_{r}$. Since in Equation (\ref{le:no_2_regular_path}) $v$ has at most $d_{v}$ neighbors in $G_r$ and $c_{*}$ is $d_{c}$-regular, $n_0 \le d_{v}(d_{c}-1)$. Considering Equation (\ref{le:path_equation}) with $l=1$, we note that $v_{1}$ has at most $d_{v}$ neighbors in $G_{r}$ and $c_{1}$ is $2$-regular. Thus, $n_{1} \le d_{v}(d_{c}-1)(d_{v}-1)$. Note that if $u$ is a variable node in $G_r$ for which the smallest integer $l$ for which Equation (\ref{le:path_equation}) holds is $l=i+1$ (where $i \in [d_{v}-2]$), then there exists a path $\mathcal{P}$ that satisfies Equation (\ref{le:path_equation}) with $v_1$ being a variable node in $G_r$ for which the smallest integer $l$ for which Equation (\ref{le:path_equation}) holds is $l=i$. Since for every $l \in [d_{v}-1]$ and every $i \in [l]$, $v_i$ has at most $d_{v}$ neighbors in $G_r$ and $c_{i}$ is $2$-regular, we have that $n_{i+1} \le (d_{v}-1)n_i$ for all $i \in [d_{v}-2]$. By induction on $i$, we get that $n_i \le d_{v}(d_{c}-1)(d_{v}-1)^{i}$ for all $i \in [d_{v}-1]$. Thus,
\begin{equation*}
|\Delta'(v)| = \displaystyle\sum\limits_{i=0}^{d_{v}-1} n_i \le \displaystyle\sum\limits_{i=0}^{d_{v}-1} d_{v}(d_{c}-1)(d_{v}-1)^{i} = d_{v}(d_{c}-1)\frac{(d_{v}-1)^{d_{v}}-1}{d_{v}-2} = q
\end{equation*}
To show that $T_{0} \le q_{0}$, note that $u \in W_{0}$ if and only if there exists a directed path from $u$ to $v_{max}$ in $G_{r}$ containing only $2$-regular check nodes. An analogous argument to the above implies that
\begin{equation*}
T_{0} \le 1 + \displaystyle\sum\limits_{i=1}^{d_{v}-1} (d_{v}-1)^{i-1} \le 1 + \frac{(d_{v}-1)^{d_{v}-1}-1}{d_{v}-2} = q_{0}
\end{equation*}
\end{proof}

\begin{corollary}\label{le:max2}
Let $G_{r}$ be the WDAG (with a single sink node) given in Lemma \ref{le:inductivelemma2} and $z_{max}$ be the maximum regular check depth of a variable node in $G_{r}$.\footnote{Note that in general $z_{max} \le r_{max}$ but the two quantities need not be equal.} Then,
\begin{equation}\label{le:sec_max_opt}{\alpha_{max} \le \max_{ (T_{0}, ..., T_{z_{max}}) \in W} { f(T_{0}, ..., T_{z_{max}}) }}\end{equation}
where: 
\begin{equation*} f(T_{0}, ..., T_{z_{max}}) = \displaystyle\sum\limits_{i=0}^{z_{max}} \frac{T_{i}}{(d_{c}-1)^{i}} \end{equation*}
and $W$ is the set of all tuples $(T_{0}, ..., T_{z_{max}}) \in \mathbb{N}^{z_{max}+1}$ satisfying the following three equations:
\begin{equation}
\displaystyle\sum\limits_{i=0}^{z_{max}} T_{i} = n_r
\end{equation}
\begin{equation}\label{le:base_ineq_sc}
T_{0} \le q_{0}
\end{equation}
\begin{equation}\label{le:inductive_ineq_sc}
\text{For all } i \in \{0,\dots,z_{max}-1\}, ~ T_{i+1} \le qT_{i}
\end{equation}
where $q =d_{v}(d_{c}-1)\frac{(d_{v}-1)^{d_{v}}-1}{d_{v}-2}$ and $q_{0} = 1 + \frac{(d_{v}-1)^{d_{v}-1}-1}{d_{v}-2}$.
\end{corollary}

\begin{proof}[\bf{Proof of Corollary \ref{le:max2}}]
The proof is similar to that of Corollary \ref{le:firstmaximization}. Setting $m=r_{max}$ in Lemma \ref{le:inductivelemma2} and noting that the leaves of $T$ have no entering flow, we get:
\begin{equation*}
\displaystyle\sum\limits_{j=1}^{n_r} \displaystyle\sum\limits_{k=1}^{y_{r_{max},j}} \Gamma_{r_{max},j,k} \geq F(U_{r_{max}}) \geq (d_{c}-1)^{r_{max}} \alpha_{max} - \displaystyle\sum\limits_{i=0}^{r_{max}-1} (d_{c}-1)^{r_{max}-i}\displaystyle\sum\limits_{j=1}^{n_r} \displaystyle\sum\limits_{k=1}^{y_{i,j}} \Gamma_{i,j,k}
\end{equation*}
Thus,
\begin{equation*}
\alpha_{max} \le \displaystyle\sum\limits_{i=0}^{r_{max}} \frac{1}{(d_{c}-1)^{i}}\displaystyle\sum\limits_{j=1}^{n_r} \displaystyle\sum\limits_{k=1}^{y_{i,j}}  \Gamma_{i,j,k}
\end{equation*}
Part \ref{le:bijection_item} of Theorem \ref{le:transformation} implies that for every $v \in V_{r}$, the regular check depth of $v$ in $G_{r}$ is equal to the minimum regular check depth in $T$ of a replicate of $v$. By parts \ref{le:add_up_item} and \ref{le:sign_conservation} of Theorem \ref{le:transformation}, we also have that for all $j \in [n_r]$, $\displaystyle\sum\limits_{i=0}^{r_{max}} \displaystyle\sum\limits_{k=1}^{y_{i,j}}  \Gamma_{i,j,k} \le 1$ and for all $i \in \{0,\dots,r_{max}\}$ and all $k \in [y_{i,j}]$, $\Gamma_{i,j,k} \le 1$ and $\{\Gamma_{i,j,k}\}_{i,k}$ all have the same sign. Thus, we get that:
\begin{equation*}
\alpha_{max} \le \displaystyle\sum\limits_{i=0}^{r_{max}} \frac{1}{(d_{c}-1)^{i}} T_i
\end{equation*}
where for every $i \in \{0,\dots,r_{max}\}$, $T_i$ is the number of variable nodes with regular check depth equal to $i$ in $G_{r}$. Since $T_{i} =0$ for all $z_{max} < i \le r_{max}$, we get that:
\begin{equation*}
\alpha_{max} \le \displaystyle\sum\limits_{i=0}^{z_{max}} \frac{1}{(d_{c}-1)^{i}} T_i
\end{equation*}
By the definitions of $T_i$ and $z_{max}$, $\displaystyle\sum\limits_{i=0}^{z_{max}} T_{i} = n_r$. The facts that $T_{i+1} \le qT_{i}$ for all $i \in \{0,\dots,z_{max}-1\}$ and $T_{0} \le q_{0}$ follow from Lemma \ref{le:q_expression}.
\end{proof} 

\begin{lemma}\label{le:upperboundcalc2}
The RHS of (\ref{le:sec_max_opt}) is $< c \times n_{r}^{1-\epsilon}$ for some constant $c > 0$ depending only on $d_{v}$ and where $0<\epsilon = \frac{\ln(d_{c}-1)}{\ln(q)}<1$.
\end{lemma}

\begin{proof}[\bf{Proof of Lemma \ref{le:upperboundcalc2}}]
Let $c = q_{0} \frac{\big(\frac{q}{d_{c}-1}\big)^2}{\frac{q}{d_{c}-1}-1}$. If $n_r \geq q_{0}$, the claim follows from Theorem \ref{le:unified_theorem} with $\lambda = q_{0}$, $\beta = q$ and $m=n_{r}$. If $n_r < q_{0}$, then the RHS of (\ref{le:sec_max_opt}) is at most $n_r < q_{0} < c$, so the claim is also true.
\end{proof}

\begin{proof}[\bf{Proof of Theorem \ref{le:maxweightsc}}]
Theorem \ref{le:maxweightsc} follows from Corollary \ref{le:max2} and Lemma \ref{le:upperboundcalc2} by noting that $|V_{r}| \le |V|$ since $V_{r} \subseteq V$ and that $\underset{e \in E}{\operatorname{max}}{~ |w(e)|} = \Omega(\underset{(v,c):w(v,c) \le 0}{\operatorname{max}}{~ |w(v,c)|})$ by the hyperflow equation (\ref{le:hyperflow_check_equation}).
\end{proof}

\section[Relation between LP decoding on a graph cover code and on a derived spatially coupled code]{\Large{\bf Relation between LP decoding on a graph cover code and on a derived spatially coupled code}}\label{le:relation_section}

\begin{definition}(Special variable nodes)\\ 
Let $\zeta$ be a graph cover code and $\zeta'$ be a fixed element of $\mathcal{D}(\zeta)$. Then, the ``special variable nodes'' of $\zeta$ are all those variable nodes that appear in $\zeta$ but not in $\zeta'$.
\end{definition}

\begin{lemma}\label{le:derivedvscoverlemma}
Let $\zeta$ be a $(d_{v},d_{c}=kd_{v},L,M)$ graph cover code and let $\zeta'$ be a be a fixed element of $\mathcal{D}(\zeta)$.\footnote{Here, $\mathcal{D}(\zeta)$ refers to Definition \ref{le:derivedpsatiallycoupledcodes}.} Let $n=(2L+1)M$ be the block length of $\zeta$ and consider transmission over the BSC. Assume $\alpha(n)$ is s.t., for any error pattern $\eta'$ on $\zeta'$, the existence of a dual witness for $\eta'$ on $\zeta'$ implies the existence of a dual witness for $\eta'$ on $\zeta'$ with maximum edge weight $< \alpha(n)$.\\ 
Then, for any error pattern $\eta'$ on $\zeta'$ and any extension $\eta$ of $\eta'$ into an error pattern on $\zeta$, the existence of a dual witness for $\eta'$ on $\zeta'$ is equivalent to the existence of a dual witness for $\eta$ on $\zeta$ with the special variable nodes having an ``extra flow'' of $d_{v} \alpha(n) + 1$.
\end{lemma}

\begin{proof}[\bf{Proof of lemma \ref{le:derivedvscoverlemma}}]
First, we prove the forward direction of the equivalence. Assume that there exists a dual witness for $\eta'$ on $\zeta'$. Then, there exists a dual witness for $\eta'$ on $\zeta'$ and with maximum edge weight $ <\alpha(n)$. This implies the existence of a dual witness for $\eta$ on $\zeta$ with the special variable nodes being source nodes and having an ``extra flow'' of $d_{v} \alpha(n) + 1$.\\ 
The reverse direction follows from the fact that given a dual witness for $\eta$ on $\zeta$, we can get a dual witness for $\eta'$ on $\zeta'$ by repeatedly removing the special variable nodes. The WDAG satisfies the LP constraints after each step since every check node in $\zeta'$ has degree $\geq 2$.
\end{proof}

\begin{corollary}\label{le:relationspatiallycover} (Relation between LP decoding on a graph cover code and on a derived spatially coupled code)\\ 
Let $\zeta$ be a $(d_{v},d_{c}=kd_{v},L,M)$ graph cover code and let $\zeta'$ be a be a fixed element of $\mathcal{D}(\zeta)$. Let $n=(2L+1)M$ be the block length of $\zeta$ and consider transmission over the BSC. Then, for any error pattern $\eta'$ on $\zeta'$ and any extension $\eta$ of $\eta'$ into an error pattern on $\zeta$, the existence of a dual witness for $\eta'$ on $\zeta'$ is equivalent to the existence of a dual witness for $\eta$ on $\zeta$ with the special variable nodes having an ``extra flow'' of $d_{v} c n^{1-\epsilon}+1$ for some $c>0$ and $0< \epsilon <1$ given in Theorem \ref{le:maxweightsc}.
\end{corollary}

\begin{proof}[{\bf Proof of Corollary \ref{le:relationspatiallycover}}]
By Theorem \ref{le:maxweightsc}, the existence of a dual witness for $\eta'$ on $\zeta'$ is equivalent to the existence of a dual witness for $\eta'$ on $\zeta'$ and with maximum edge weight $< c n^{1-\epsilon}$ for some $c>0$. Plugging this expression in Lemma \ref{le:derivedvscoverlemma}, we get the statement of Corollary \ref{le:relationspatiallycover}.
\end{proof}

\section[Interplay between crossover probability and LP excess]{\Large{\bf Interplay between crossover probability and LP excess}}\label{le:interplay_section}
In this section, we show that if the probability of LP decoding success is large on some BSC, then if we slightly decrease the crossover probability of the BSC, we can find a dual witness with a non-negligible ``gap'' in the inequalities (\ref{le:dw_var_equation}) with high probability.
\begin{theorem}\label{le:interplaytheorem} (Interplay between crossover probability and LP excess)\\ 
Let $\zeta$ be a binary linear code with Tanner graph $(V,C,E)$ where $V=\{v_{1}, \cdots ,v_{n}\}$. Let $\epsilon,\delta>0$ and $\epsilon' = \epsilon + (1-\epsilon)\delta$. Assume that $\epsilon,\epsilon',\delta <1$. Let $q_{\epsilon'}$ be the probability of LP decoding error on the $\epsilon'$-BSC. For every error pattern $x \in \{0,1\}^{n}$, if $G=(V,C,E,w,\gamma)$ is a WDAG corresponding to a dual witness for $x$, let $f(w) \in \mathbb{R}^{n}$ be defined by
\begin{equation}\label{le:f_function_definition}
f_{i}(w) = \displaystyle\sum\limits_{c \in N(v_{i}):w(v_{i},c)>0} w(v_{i},c) -\displaystyle\sum\limits_{c \in N(v_{i}):w(v_{i},c) \le 0} (-w(v_{i},c)) = \displaystyle\sum\limits_{c \in N(v_{i})} w(v_{i},c)
\end{equation}
for all $i \in [n]$. Then,
\begin{equation*}
Pr_{x \sim Ber(\epsilon,n)}\{\exists \text{ a dual witness $w$ for $x$ s.t. } f_{i}(w) < \gamma(v_{i}) - \frac{\delta}{2} \text{, } \forall i \in [n]\} \ge 1-\frac{2q_{\epsilon'}}{\delta}
\end{equation*}
In other words, if we let $\gamma(v_{i})-f_{i}(w)$ be the ``LP excess'' on variable node $i$, then the probability (over the $\epsilon$-BSC) that there exists a dual witness with LP excess at least $\delta/2$ on all the variable nodes is at least $1-\frac{2q_{\epsilon'}}{\delta}$.
\end{theorem}

\begin{proof}[{\bf Proof of Theorem \ref{le:interplaytheorem}}]
Decompose the $\epsilon'$-BSC into the bitwise OR of the $\epsilon$-BSC and the $\delta$-BSC as follows. Let $x \sim Ber(\epsilon,n)$, $e'' \sim Ber(\delta,n)$ and $e = x \lor e''$. Hence, $e \sim Ber(\epsilon',n)$. For every $x \in \{0,1\}^n$, we will construct a dual witness $w^{x}$ with excess $\delta/2$ on all variable nodes by averaging and scaling the dual witnesses of $x \lor e''$ where $e'' \sim Ber(\delta,n)$. More precisely, for every $x \in \{0,1\}^n$, let $w^{x} = \frac{(1+\frac{\delta}{2})}{(1-\frac{\delta}{2})}E_{e''\sim Ber(\delta,n)}\{v^{x \lor e''}\}$ where $v^{x}$ is an arbitrary dual witness for $x$ if $x$ has one and $v^{x}$ is the zero vector otherwise. Note that $w^{x}$ always satisfies the check node constraints, i.e. for any $x \in \{0,1\}^{n}$, any $c \in C$ and any $v,v' \in V$, we have $w^{x}(v,c) + w^{x}(v',c) \ge 0$. We now show that, with probability at least $1-\frac{2q_{\epsilon'}}{\delta}$ over $x \sim Ber(\epsilon,n)$, $w^{x}$ satisfies (\ref{le:dw_var_equation}) with LP excess at least $\delta/2$ on all variable nodes. For any weight function $w:V \times C \to \mathbb{R}$ on the Tanner graph $(V,C,E)$, we define $f(w)$ by Equation (\ref{le:f_function_definition}). For every $x \in \{0,1\}^{n}$, define the event $L^{x} =\{ x \text{ has a dual witness}\}$ and define $\tilde{x}$ by $\tilde{x_{i}} = (-1)^{x_{i}}$ for all $i \in [n]$. We have that:
\begin{align*}
f(w^{x})&=\frac{(1+\frac{\delta}{2})}{(1-\frac{\delta}{2})}E_{e''\sim Ber(\delta,n)}\{f(w^{x \lor e''})\}\\ 
&= \frac{(1+\frac{\delta}{2})}{(1-\frac{\delta}{2})}\Big(E_{e''\sim Ber(\delta,n)}\{f(w^{x \lor e''})|L^{x \lor e''} \}Pr_{e''\sim Ber(\delta,n)}\{L^{x \lor e''}\} \\ 
&+ E_{e''\sim Ber(\delta,n)}\{f(w^{x \lor e''})|\overline{L^{x \lor e''}}\}Pr_{e''\sim Ber(\delta,n)}\{\overline{L^{x \lor e''}}\}\Big)\\ 
&= \frac{(1+\frac{\delta}{2})}{(1-\frac{\delta}{2})}E_{e''\sim Ber(\delta,n)}\{f(w^{x \lor e''})|L^{x \lor e''}\}Pr_{e''\sim Ber(\delta,n)}\{L^{x \lor e''}\} ~ (\text{since }E_{e''\sim Ber(\delta,n)}\{f(w^{x \lor e''})|\overline{L^{x \lor e''}}\} =0)\\ 
&\le \frac{(1+\frac{\delta}{2})}{(1-\frac{\delta}{2})}E_{e''\sim Ber(\delta,n)}\{\widetilde{x \lor e''}|L^{x \lor e''} \}Pr_{e''\sim Ber(\delta,n)}\{L^{x \lor e''}\} ~ (\text{by equation (\ref{le:dw_var_equation})})\\ 
&= \frac{(1+\frac{\delta}{2})}{(1-\frac{\delta}{2})}\Big(E_{e''\sim Ber(\delta,n)}\{\widetilde{x \lor e''}\} - E_{e''\sim Ber(\delta,n)}\{\widetilde{x \lor e''}|\overline{L^{x \lor e''}}\} \times \phi_{x}\Big)
\end{align*}
where $\phi_{x} = Pr_{e''\sim Ber(\delta,n)}\big\{\overline{L^{x \lor e''}}\big\}$.
Note that for every $i \in [n]$, we have:
\[
 \bigg(E_{e''\sim Ber(\delta,n)}\{\widetilde{x \lor e''}\}\bigg)_{i} = \begin{dcases*}
        -1  & if $x_{i}=1$.\\ 
        \delta(-1)+(1-\delta)(+1) = 1-2\delta & if $x_{i}=0$.
        \end{dcases*}
\]
Moreover, $E_{e''\sim Ber(\delta,n)}\{\widetilde{x \lor e''}|\overline{L^{x \lor e''}}\} \ge -1$ since every coordinate of $\widetilde{x \lor e''}$ is $\ge -1$. Therefore,\\ 
\[
 f_{i}(w^{x}) \le \begin{dcases*}
        \frac{(1+\frac{\delta}{2})}{(1-\frac{\delta}{2})}(-1+\phi_{x})  & if $x_{i}=1$.\\ 
        \frac{(1+\frac{\delta}{2})}{(1-\frac{\delta}{2})}(1-2\delta+\phi_{x}) & if $x_{i}=0$.
        \end{dcases*}
\]
We now find an upper bound on $\phi_{x}$. Note that $\phi_{x}$ is a non-negative random variable with mean\\ 
\begin{align*}
E_{x \sim Ber(\epsilon,n)}\{\phi_{x}\} &= E_{x \sim Ber(\epsilon,n)}\big\{Pr_{e''\sim Ber(\delta,n)}\{\overline{L^{x \lor e''}}\}\big\} = Pr_{x \sim Ber(\epsilon,n),e''\sim Ber(\delta,n)}\big\{\overline{L^{x \lor e''}}\big\}\\ 
&= Pr_{e \sim Ber(\epsilon',n)}\big\{\overline{L^{e}}\big\} = q_{\epsilon'} ~ ~ ~ \text{(by Theorem \ref{le:existencedualwitness})}
\end{align*}
By Markov's inequality, $Pr_{x \sim Ber(\epsilon,n)}\{\phi_{x} \ge \frac{\delta}{2}\} \le \frac {E_{x \sim Ber(\epsilon,n)}\{\phi_{x}\}}{\frac{\delta}{2}} = \frac{2q_{\epsilon'}}{\delta}$. Thus, the probability over $x \sim Ber(\epsilon,n)$ that for all $i \in [n]$, $f_{i}(w^{x}) < \frac{(1+\frac{\delta}{2})}{(1-\frac{\delta}{2})}(-1+\frac{\delta}{2})$ if $x_{i}=1$ and $f_{i}(w^{x}) < \frac{(1+\frac{\delta}{2})}{(1-\frac{\delta}{2})}(1-\frac{3\delta}{2})$ if $x_{i}=0$, is at least
\begin{equation*}
Pr_{x \sim Ber(\epsilon,n)}\{\phi_{x} < \frac{\delta}{2}\} = 1-Pr_{x \sim Ber(\epsilon,n)}\{\phi_{x} \ge \frac{\delta}{2}\} \ge 1-\frac{2q_{\epsilon'}}{\delta}
\end{equation*}
Note that for all $0 \le \delta < 1$, we have that $\frac{(1+\frac{\delta}{2})}{(1-\frac{\delta}{2})}(1-\frac{3\delta}{2}) \le 1-\frac{\delta}{2}$. Thus, the probability over $x \sim Ber(\epsilon,n)$ that $f_{i}(w^{x}) < (-1)^{x_i}-\frac{\delta}{2}$ for all $i \in [n]$, is at least $1-\frac{2q_{\epsilon'}}{\delta}$. So we conclude that
\begin{equation*}
Pr_{x \sim Ber(\epsilon,n)}\{\exists \text{ a dual witness $w$ for $x$ s.t. } f_{i}(w) < \gamma(v_{i}) - \frac{\delta}{2} \text{, } \forall i \in [n]\} \ge 1-\frac{2q_{\epsilon'}}{\delta}
\end{equation*}
\end{proof}

\section[$\bf{\xi_{GC} = \xi_{SC}}$]{\Large{\bf $\xi_{GC} = \xi_{SC}$}}\label{le:proof_main_result_section}
In this section, we use the results of Sections \ref{le:maxweightsc_section}, \ref{le:relation_section} and \ref{le:interplay_section} to prove the main result of the paper which is restated below.
\begin{theorem}\label{le:equalitytheorem} (Main result: $\xi_{GC} = \xi_{SC}$)\\ 
Let $\Gamma_{GC}$ be a $(d_{v},d_{c}=kd_{v},L,M)$ graph cover ensemble with $d_{v}$ an odd integer and $M$ divisible by $k$. Let $\Gamma_{SC}$ be the $(d_{v},d_{c}=kd_{v},L-\hat{d_{v}},M)$ spatially coupled ensemble which is sampled by choosing a graph cover code $\zeta \sim \Gamma_{GC}$ and returning a element of $\mathcal{D}(\zeta)$ chosen uniformly at random\footnote{Here, $\mathcal{D}(\zeta)$ refers to Definition \ref{le:derivedpsatiallycoupledcodes}.}. Denote by $\xi_{GC}$ and $\xi_{SC}$ the respective LP threholds of $\Gamma_{GC}$ and $\Gamma_{SC}$ on the BSC. There exists $\nu>0$ depending only on $d_{v}$ and $d_{c}$ s.t. if $M = o(L^{\nu})$ and $\Gamma_{SC}$ satisfies the property that for any constant $\Delta > 0$, 
\begin{equation}\label{le:needed_condition_repeated}
Pr_{\zeta' \sim \Gamma_{SC} \atop (\xi_{SC} - \Delta)\text{-}BSC}[\text{LP error on }\zeta'] = o(\frac{1}{L^2})
\end{equation}
Then, $\xi_{GC} = \xi_{SC}$.
\end{theorem}

\begin{lemma}\label{le:ensembles_lemma}
Assume that the ensemble $\Gamma_{SC}$ satisfies the property (\ref{le:needed_condition_repeated}) for every constant $\Delta > 0$. Then, for all constants $\Delta_{1},\Delta_{2}, \alpha, \beta> 0$, there exists a graph cover code $\zeta \in \Gamma_{GC}$, with derived spatially coupled codes $\zeta'_{-L}, \dots ,\zeta'_{L}$, satisfying the following two properties for sufficiently large $L$:
\begin{enumerate}
\item $Pr_{(\xi_{GC}+\Delta_{2})\text{-}BSC}[\text{LP decoding success on }\zeta] \le \alpha$.
\item For all $i \in [-L:L]$, $Pr_{(\xi_{SC}-\Delta_{1})\text{-}BSC}[\text{LP decoding error on }\zeta'_{i}] \le \beta/(2L+1)$.
\end{enumerate}
\end{lemma}

\begin{proof}[{\bf Proof of lemma \ref{le:ensembles_lemma}}]
Note that a random code $\zeta \sim \Gamma_{GC}$ satisfies the $2$ properties above with high probability:
\begin{align*}
&Pr_{\zeta \sim \Gamma_{GC}}\big[Pr_{(\xi_{GC}+\Delta_{2})\text{-}BSC}[\text{Success on }\zeta] > \alpha \text{ or }\exists i \in [-L:L] \text{ s.t. }Pr_{(\xi_{SC}-\Delta_{1})\text{-}BSC}[\text{Error on }\zeta'_{i}] > \beta(2L+1)\big]\\ 
&\le \frac{1}{\alpha}Pr_{\zeta \sim \Gamma_{GC} \atop (\xi_{GC}+\Delta_{2})\text{-}BSC}[\text{LP decoding success on }\zeta] + \frac{(2L+1)^2}{\beta}Pr_{\zeta' \sim \Gamma_{SC} \atop (\xi_{SC}-\Delta_{1})\text{-}BSC}[\text{LP decoding error on }\zeta']\\ 
&=o(1)
\end{align*}
Note that the inequality above follows from Markov's inequality and the union bound. We conclude that there exists a graph cover code $\zeta \in \Gamma_{GC}$ satisfying the $2$ properties above.
\end{proof}

\begin{lemma}\label{le:greaterthanorequallemma}
$\xi_{GC} \ge \xi_{SC}$
\end{lemma}

\begin{proof}[{\bf Proof of lemma \ref{le:greaterthanorequallemma}}]
We proceed by contradiction. Assume that $\xi_{GC} < \xi_{SC}$. Let:
\begin{align*}
&\delta = (\xi_{SC}-\xi_{GC})/2\\ 
&\eta = \xi_{SC} - \delta\\ 
&\lambda = \eta - \delta/2 = \xi_{GC} + \delta/2
\end{align*}
Note that $\eta > \lambda + (1-\lambda)\delta/2$. Let $\zeta$ be one of the graph cover codes whose existence is guaranteed by Lemma \ref{le:ensembles_lemma} with $\Delta_{1} = \delta$, $\Delta_{2} = \delta/2$ and $\alpha,\beta >0$ with $\alpha < 1 - 2\beta/\delta$ and let $\zeta'_{-L}, \dots ,\zeta'_{L}$ be the spatially coupled codes that are derived from $\zeta$. Let $\mu$ be an error pattern on $\zeta$ and let $\mu_{i}$ be the restriction of $\mu$ to $\zeta'_{i}$ for every $i \in [-L:L]$. Define the event:
\begin{equation*}
E_1 = \{ \forall i \in [-L:L], \exists\text{ a dual witness for $\mu_{i}$ on $\zeta^{'}_{i}$ with excess $\delta/2$ on all variable nodes}\}
\end{equation*}
Then,
\begin{equation*}
\overline{E_1} = \{ \exists i \in [-L:L]\text{ s.t. } \nexists \text{ a dual witness for $\mu_{i}$ on $\zeta^{'}_{i}$ with excess $\delta/2$ on all variable nodes}\}
\end{equation*} 
Thus, \begin{align*} Pr_{\lambda\text{-}BSC}\{ \overline{E_1}\} &\le \displaystyle\sum\limits_{i=-L}^L Pr_{\lambda\text{-}BSC} \{ \nexists \text{ a dual witness for $\zeta^{'}_{i}$ with excess $\delta/2$ on all variable nodes}\}\\ &\le \displaystyle\sum\limits_{i=-L}^L \frac{2}{\delta} Pr_{\eta\text{-}BSC} \{ \text{LP decoding error on $\zeta^{'}_{i}$}\} \text{ (by Theorem \ref{le:interplaytheorem})}\\ &\le \displaystyle\sum\limits_{i=-L}^L \frac{2}{\delta} \times \frac{\beta}{2L+1} = \frac{2 \beta}{\delta} \end{align*}
If event $E_1$ is true, then by Corollary \ref{le:relationspatiallycover}, for every $l \in [-L:L]$, there exists a dual witness $\{\tau_{ij}^l ~ | ~ i \in V, j \in C \}$ for $\mu$ on $\zeta$ with the special variable nodes being at positions $[l,l+2\hat{d_{v}}-1]$ and having an ``extra flow'' of $d_{v} c n^{1-\epsilon}+1$ with $c>0$ and $\epsilon > 0$ given in Theorem \ref{le:maxweightsc} and with the non-special variable nodes having excess $\frac{\delta}{2}$. Then, we can construct a dual witness for $\mu$ on the graph cover code $\zeta$ (with no extra flows) by averaging the above $2L+1$ dual witnesses as follows. For every $i \in V$ and every $j \in C$, let:
\begin{equation*}
\tau_{ij}^{avg} = \frac{1}{2L+1}\displaystyle\sum\limits_{l=-L}^L \tau_{ij}^l
\end{equation*}
We claim that $\{\tau_{ij}^{avg}\}_{i,j}$ forms a dual witness for $\mu$ on $\zeta$. In fact, for each $i \in V$, $j \in C$ and $l \in [-L:L]$, $\tau_{ij}^l + \tau_{i'j}^l \geq 0$ which implies that:
\begin{equation*}
\tau_{ij}^{avg} + \tau_{i'j}^{avg} = \frac{1}{2L+1}\displaystyle\sum\limits_{l=-L}^L (\tau_{ij}^l + \tau_{i'j}^l) \ge 0
\end{equation*}
Moreover, for all $i \in V$, we have that:
\begin{align*} \displaystyle\sum\limits_{j \in N(i)} \tau_{ij}^{avg} &= \displaystyle\sum\limits_{j \in N(i)} \Big(\frac{1}{2L+1}\displaystyle\sum\limits_{l=-L}^L \tau_{ij}^l \Big)\\  &= \frac{1}{2L+1}\displaystyle\sum\limits_{l=-L}^L \Big(\displaystyle\sum\limits_{j \in N(i)} \tau_{ij}^l \Big)\\ &< \frac{1}{2L+1} \Big( (d_{v}-1)(d_v c (M(2L+1))^{1-\epsilon} + 1 + \gamma_i) + (2L+1 - (d_{v}-1))(\gamma_i - \frac{\delta}{2}) \Big)\\ &=\gamma_i + (d_{v}-1)d_v c \frac{(M(2L+1))^{1-\epsilon}}{2L+1} + \frac{(d_{v}-1)\delta}{2(2L+1)} + \frac{d_{v}-1}{2L+1} -\frac{\delta}{2} \\ &< \gamma_i \text{ if }M=o(L^{\nu}) \text{, $L$ sufficiently large and }\nu = \epsilon/(1-\epsilon) \end{align*}
Since $Pr_{\lambda\text{-}BSC}\{ \text{LP decoding success on }\zeta\} \ge Pr_{\lambda\text{-}BSC}\{ E_1\} = 1 - Pr_{\lambda\text{-}BSC}\{ \overline{E_1}\}$, then,
\begin{align*}
Pr_{\lambda\text{-}BSC}\{ \text{LP decoding success on }\zeta\} \geq 1-\frac{2 \beta}{\delta}
\end{align*}
which contradicts the fact that:
\begin{equation*}
Pr_{\lambda\text{-}BSC}[\text{LP decoding success on }\zeta] = Pr_{(\xi_{GC}+\Delta_{2})\text{-}BSC}[\text{LP decoding success on }\zeta] \le \alpha < 1-\frac{2 \beta}{\delta}
\end{equation*}

\end{proof}

\begin{lemma}\label{le:lessthanorequal}
$\xi_{GC} \le \xi_{SC}$
\end{lemma}

\begin{proof}[{\bf Proof of Lemma \ref{le:lessthanorequal}}]
Let $\zeta$ be a graph cover code and $D(\zeta)$ be the set of all derived spatially coupled codes of $\zeta$. Let $\mu$ be an error pattern on $\zeta$ and $\mu'$ be the restriction of $\mu$ to $\zeta'$ for some $\zeta' \in D(\zeta)$. Given a dual witness for $\mu$ on $\zeta$, we can get a dual witness for $\mu'$ on $\zeta'$ by repeatedly removing the special variable nodes of $\zeta$. Note that the dual witness is maintained after each step since every check node in $\zeta'$ has degree $\geq 2$. So if there is LP decoding success for $\eta$ on $\zeta$, then for every $\zeta' \in D(\zeta)$, there is LP decoding success for $\eta'$ on $\zeta'$, where $\eta'$ is the restriction of $\eta$ to $\zeta'$. Therefore, for every $\epsilon >0$ and every $\zeta' \in D(\zeta)$, we have that:
\begin{equation*}
Pr_{\epsilon\text{-}BSC}[\text{LP decoding error on }\zeta'] \le Pr_{\epsilon\text{-}BSC}[\text{LP decoding error on }\zeta]
\end{equation*}
This implies that for every $\epsilon >0$, we have that:
\begin{equation*}
Pr_{\zeta' \sim \Gamma_{SC} \atop \epsilon\text{-}BSC}[\text{LP decoding error on }\zeta'] \le Pr_{\zeta \sim \Gamma_{GC} \atop \epsilon\text{-}BSC}[\text{LP decoding error on }\zeta]
\end{equation*}
So we conclude that $\xi_{GC} \le \xi_{SC}$.
\end{proof}

\begin{proof}[{\bf Proof of Theorem \ref{le:equalitytheorem}}]
Theorem \ref{le:equalitytheorem} follows from Lemma \ref{le:greaterthanorequallemma} and Lemma \ref{le:lessthanorequal}.
\end{proof}

\appendix

\section[Appendix]{\Large{\bf Appendix}}

\subsection{Proof of Theorem \ref{le:existencedualwitness}}\label{le:app_ex_dual_witness}

The goal of this section is to prove Theorem \ref{le:existencedualwitness} which is restated below.

\begin{reptheorem}{le:existencedualwitness} (Existence of a dual witness and LP decoding success) \\ 
Let $\mathcal{T}=(V,C,E)$ be a Tanner graph of a binary linear code with block length $n$ and let $\eta \in \{0,1\}^{n}$ be any error pattern. Then, there is LP decoding success for $\eta$ on $\mathcal{T}$ if and only if there is a dual witness for $\eta$ on $\mathcal{T}$.
\end{reptheorem}
\noindent Note that the ``if'' part of the statement was proved in \cite{feldman2007lp}. The argument below establishes both directions.
We first state some definitions and prove some facts from convex geometry that will be central to the proof of Theorem \ref{le:existencedualwitness}.
\begin{definition}\label{le:topologydefinitions}
Let $S$ be a subset of $\mathbb{R}^{n}$. The convex span of $S$ is defined to be $conv(S) = \{ \alpha x + (1-\alpha)y ~ | ~ x,y \in S \text{ and } \alpha \in [0,1] \}$. The conic span of $S$ is defined to be $cone(S) = \{ \alpha x + \beta y ~ | ~ x,y \in S \text{ and } \alpha,\beta \in \mathbb{R}_{\geq 0} \}$. The set $S$ is said to be convex if $S = conv(S)$ and $S$ is said to be a cone if $S = cone(S)$. Also, $S$ is said to be a convex polyhedron if $S = \{ x \in \mathbb{R}^{n} ~ | ~ Ax \geq b\}$ for some matrix $A \in \mathbb{R}^{m \times n}$ and some $b \in \mathbb{R}^{n}$ and $S$ is said to be a polyhedral cone if $S$ is both a convex polyhedron and a cone. The interior of $S$ is denoted by $int(S)$ and the closure of $S$ is denoted by $cl(S)$.\\ 
\comments{
\item For any point $x \in \mathbb{R}^{n}$ and any $r > 0$, the neighborhood $N_{r}(x)$ is defined to be $N_{r}(x) = \{ y \in \mathbb{R}^{n} ~ | ~ \| x-y \|_{2} < r \}$ where $\|.\|_{2}$ denotes the Eucledian norm\footnote{The Euclidean norm of a vector $u \in \mathbb{R}^{n}$ is given by $\|u\|_{2} = \sqrt{\displaystyle\sum\limits_{i = 1}^{n} u_{i}^2}$ }.
\item The interior of $S$ is defined to be $int(S) = \{ x \in S ~ | ~ \exists r > 0 \text{ s.t. } N_{r}(x) \subseteq S \}$.
\item The closure of $S$ is defined to be $cl(S) = S \cup S'$ where $S'$ is the set of all limit points of $S$ which is given by $S' = \{x \in \mathbb{R}^{n} ~ | ~ \forall r > 0, \exists y \in $S$ \text{ s.t. } y \neq x \text{ and } y \in N_{r}(x) \}$.
}
Let $K$ be a polyhedral cone of the form $K = \{x \in \mathbb{R}^{n} ~ | ~ Ax \geq 0 \}$ for some matrix $A \in \mathbb{R}^{m \times n}$. For any $x \in K$ s.t. $x \neq 0$, the ray of $K$ in the direction of $x$ is defined to be the set $R(x) = \{ \lambda x ~ | ~ \lambda \geq 0 \}$. A ray $R(x)$ of $K$ is said to be an extreme ray of $K$ if for any $y,z \in \mathbb{R}^{n}$ and any $\alpha, \beta \geq 0$, $R(x) = \alpha R(y) + \beta R(z)$ implies that $y,z \in R(x)$.
\end{definition}

\begin{lemma}\label{le:topologyneededclaim}
If $S$ is a convex subset of $\mathbb{R}^{n}$, then $int\big(({\mathbb{R}_{\geq 0}})^{n} + S\big) = ({\mathbb{R}_{> 0}})^{n} + S$.
\end{lemma}

\begin{proof}[{\bf Proof of Lemma \ref{le:topologyneededclaim}}]
For all $\alpha \in ({\mathbb{R}_{> 0}})^{n} + S$, $\alpha = r + s$ where $r \in ({\mathbb{R}_{> 0}})^{n}$ and $s \in S$. Thus, the ball centered at $\alpha$ and of radius $\min_{i \in [n]} {r_{i}} > 0$ is contained in $\big(({\mathbb{R}_{\geq 0}})^{n} + S\big)$. Hence, $\alpha \in int\big(({\mathbb{R}_{\geq 0}})^{n} + S\big)$. Therefore, $({\mathbb{R}_{> 0}})^{n} + S \subseteq int\big(({\mathbb{R}_{\geq 0}})^{n} + S\big)$.\\ 
Conversely, for all $\alpha \in int\big(({\mathbb{R}_{\geq 0}})^{n} + S\big)$, $\alpha = r + s$ where $r \in ({\mathbb{R}_{\geq 0}})^{n}$ and $s \in S$. Moreover, since $\alpha \in int\big(({\mathbb{R}_{\geq 0}})^{n} + S\big)$, there exists $u \in ({\mathbb{R}_{> 0}})^{n}$ s.t. $\alpha + u \in \big(({\mathbb{R}_{\geq 0}})^{n} + S\big)$ and $\alpha - u \in \big(({\mathbb{R}_{\geq 0}})^{n} + S\big)$. Note that $\alpha + u = r + u + s$ and that $\alpha - u = r' + s'$ for some $r' \in ({\mathbb{R}_{\geq 0}})^{n}$ and $s' \in S$. Thus, $\alpha = \frac{(\alpha+u)+(\alpha-u)}{2} = \frac{r+u+r'}{2} + \frac{s+s'}{2} = r'' + s''$ where $r'' = \frac{r+u+r'}{2} \in ({\mathbb{R}_{> 0}})^{n}$ and $s'' = \frac{s+s'}{2} \in S$ since $S$ is a convex set. Hence, $int\big(({\mathbb{R}_{\geq 0}})^{n} + S\big) \subseteq ({\mathbb{R}_{> 0}})^{n} + S$.\\ 
Therefore, $int\big(({\mathbb{R}_{\geq 0}})^{n} + S\big) = ({\mathbb{R}_{> 0}})^{n} + S$.
\end{proof}

\begin{lemma}\label{le:intermediatelemma}
Let $S_{1},..,S_{p}$ be finite subsets of ${\mathbb{R}}^{n}$ each containing the zero vector. Then, \begin{equation*} cone\big(\displaystyle\bigcap\limits_{j=1}^{p} conv(S_{j})\big) = \displaystyle\bigcap\limits_{j=1}^{p} cone(S_{j}). \end{equation*}
\end{lemma}

\begin{proof}[{\bf Proof of Lemma \ref{le:intermediatelemma}}]
Clearly, $cone\big(\displaystyle\bigcap\limits_{j=1}^{p} conv(S_{j})\big) \subseteq \displaystyle\bigcap\limits_{j=1}^{p} cone(S_{j})$. To prove the other direction, we first note that $0 \in cone\big(\displaystyle\bigcap\limits_{j=1}^{p} conv(S_{j})\big)$. For any non-zero $x \in \displaystyle\bigcap\limits_{j=1}^{p} cone(S_{j})$, we have that for all $j \in [p]$, $x = \displaystyle\sum\limits_{s \in S_j} a_{s,j}s$ where for any $s \in S_{j}, ~ a_{s,j} \geq 0$. Let $j_{max} = \underset{j \in [p]}{\operatorname{argmax}}{~ \displaystyle\sum\limits_{s \in S_j} a_{s,j}}$. Since $x \neq 0, ~ D = \displaystyle\sum\limits_{s \in S_{j_{max}}} a_{s,j_{max}} > 0$. Thus, for any $j \in [p]$, we have $\frac{x}{D} = \displaystyle\sum\limits_{s \in S_j} \big( \frac{a_{s,j}}{D} \big) s + \big(1 - \displaystyle\sum\limits_{s \in S_j} \frac{a_{s,j}}{D} \big)0$. Since for all $j \in [p]$, $0 \leq \displaystyle\sum\limits_{s \in S_j} a_{s,j} \leq D$ and $0 \in S_{j}$, we conclude that $\frac{x}{D} \in conv(S_{j})$ for all $j \in [p]$. Hence, $x \in cone\big(\displaystyle\bigcap\limits_{j=1}^{p} conv(S_{j})\big)$. Therefore, $\displaystyle\bigcap\limits_{j=1}^{p} cone(S_{j}) \subseteq cone\big(\displaystyle\bigcap\limits_{j=1}^{p} conv(S_{j})\big)$.
\end{proof}


\begin{lemma}\label{le:rayslemma}
Let $K$ be a polyhedral cone of the form $K = \{x \in \mathbb{R}^{m} ~ | ~ Ax \geq 0 \}$ for some matrix $A \in \mathbb{R}^{l \times m}$ of rank $m$. For any $x \in K$ s.t. $x \neq 0$, we have:
\begin{enumerate}
\item\label{le:first_part_rays} If $R(x)$ is an extreme ray of $K$, then there exists an $(m-1)\times m$ submatrix $A'$ of $A$ s.t. the rows of $A'$ are linearly independent and $A' x =0$.
\item\label{le:second_part_rays} $K = cone(R)$ where $R = \underset{\text{extreme rays }R(x)\text{ of }K} \bigcup R(x)$.
\end{enumerate}
\end{lemma}

\begin{proof}[{\bf Proof of Lemma \ref{le:rayslemma}}]
See Section $8.8$ of \cite{schrijver1998theory}.
\end{proof}

\begin{lemma}\label{le:equalitylemma}
For all $m \ge 2$, we have that
\begin{equation*}
\big\{ y \in ({\mathbb{R}_{\geq 0}})^{m} ~ | ~ \displaystyle\sum\limits_{i = 1,~ i \neq i_{0}}^{m} y_{i} \ge y_{i_{0}}, \forall i_{0} \in [m] \big\} = cone \{ z \in \{0,1\}^{m} ~ | ~ w(z)=2\}
\end{equation*}
\end{lemma}

\begin{proof}[{\bf Proof of Lemma \ref{le:equalitylemma}}]
Let $K_{m} = \big\{ y \in ({\mathbb{R}_{\geq 0}})^{m} ~ | ~ \displaystyle\sum\limits_{i = 1,~ i \neq i_{0}}^{m} y_{i} \ge y_{i_{0}}, \forall i_{0} \in [m] \big\}$ and $X_{m} = cone \{ z \in \{0,1\}^{m} ~ | ~ w(z)=2\}$. Clearly, $X_{m} \subseteq K_{m}$. We now prove that $K_{m} \subseteq X_{m}$. Note that $K_{m}$ can be written in the following form:
\begin{align*}
K_{m} &= \big\{ y \in {\mathbb{R}}^{m} ~ | ~ y_{i} \geq 0 ~ \forall i \in [m] \text{ and } \displaystyle\sum\limits_{i = 1,~ i \neq i_{0}}^{m} y_{i} \ge y_{i_{0}}, \forall i_{0} \in [m] \big\}\\ 
&=\{y \in \mathbb{R}^{m} ~ | ~ Ay \geq 0 \} \text{ where }A \in \mathbb{R}^{2m \times m} \text{ has rank }m
\end{align*}
By part \ref{le:second_part_rays} of Lemma \ref{le:rayslemma}, we then have:
$K_{m} = cone(R)$ where $R = \underset{\text{extreme rays }R(y)\text{ of }K_{m}} \bigcup R(y)$. Therefore, by part \ref{le:first_part_rays} of Lemma \ref{le:rayslemma}, it is sufficient to show that if $y \in {\mathbb{R}}^{m}$ satisfies any $(m-1)$ equations of $K_{m}$ with equality, then $y$ should be an element of $cone \{ z \in \{0,1\}^{m} ~ | ~ w(z)=2\}$. Note that we have two types of equations:
\begin{enumerate}
\item[(I)] $\displaystyle\sum\limits_{i = 1,~ i \neq i_{0}}^{m} y_{i} - y_{i_{0}} = 0$ for some $i_{0} \in [m]$.
\item[(II)] $y_{i} = 0$ for some $i \in [m]$.
\end{enumerate}
Consider any $(m-1)$ equations of $K_{m}$, satisfied with equality. We distinguish two cases:\\ 
Case 1: At least $(m-2)$ of those equations are of Type (II). Without loss of generality, we can assume that $y_{i} = 0$ for all $i \in \{3,\dots,m\}$. Moreover, since $y \in K_{m}$, we have that $y_{1} - y_{2} \geq 0$ and $y_{2} - y_{1} \geq 0$, which implies that $y_{1} = y_{2}$. Therefore, we conclude that $y = y_{1} (1 ~ 1 ~ 0 ~ \dots ~ 0)^{T} \in X_{m}$.\\ 
Case 2: At most $(m-3)$ equations are of Type (II). Hence, at least $2$ equations are of Type (I). Without loss of generality, we can assume that $\displaystyle\sum\limits_{i = 1,~ i \neq 1}^{m} y_{i} = y_{1}$ and $\displaystyle\sum\limits_{i = 1,~ i \neq 2}^{m} y_{i} = y_{2}$. Adding up the last $2$ equations, we get $\displaystyle\sum\limits_{i = 3}^{m} y_{i} = 0$. Since $y \in K_{m}$, we have $y_{i} \geq 0$ for all $i \in \{3,\dots,m \}$. Therefore, we get $y_{i} = 0$ for all $i \in \{3,\dots,m \}$. Similarily to Case 1 above, this implies that $y \in X_{m}$.
\end{proof}

\begin{proof}[{\bf Proof of Theorem \ref{le:existencedualwitness}}]
The ``fundamental polytope'' $P$ considered by the LP decoder was introduced by \cite{koetter2003graph} and is defined by $P = \underset{j \in C}\bigcap conv(C_{j})$ where $C_{j}= \{ z \in \{0,1\}^{n}: \text{ $w(z|_{N(j)})$ is even}\}$ for any $j \in C$. For any error pattern $\eta \in \{0,1\}^{n}$, let $\widetilde{\eta} \in \{-1,1\}^{n}$ be given by $\widetilde{\eta}_{i} = (-1)^{\eta_{i}}$ for all $i \in [n]$. Also, for any $x,y \in \mathbb{R}^n$, let their inner product be $\langle x , y \rangle = \displaystyle\sum\limits_{i = 1}^{n} x_{i} y_{i}$. Then, under the all zeros assumption, there is LP decoding success for $\eta$ on $\zeta$ if and only if the zero vector is the unique optimal solution to the LP (\ref{le:relaxed_LP}), i.e. if and only if $\langle \widetilde{\eta},0 \rangle < \langle \widetilde{\eta},y \rangle$ for every non-zero $y \in P$, which is equivalent to $\widetilde{\eta} \in int(P^{*}) = int(\mathcal{K}^{*})$ where $\mathcal{K}= cone\{P\}$ is the ``fundamental cone'' and for any $S \subseteq \mathbb{R}^n$, the dual $S^{*}$ of $S$ is given by $S^{*}=\{z \in \mathbb{R}^{n} ~ | ~ \langle z,x \rangle \geq 0 ~ \forall x \in S\}$. By Lemmas \ref{le:intermediatelemma} and \ref{le:equalitylemma}, we have:
\begin{align*}
\mathcal{K} & = cone\big(\underset{j \in C} \bigcap conv(C_{j}) \big) = \underset{j \in C} \bigcap cone(C_{j}) = \underset{j \in C}  \bigcap cone \{ z \in \{0,1\}^{n} ~ | ~ w(z|_{N(j)}) \text{ is even}\}\\ 
&= \underset{j \in C} \bigcap cone \{ z \in \{0,1\}^{n} ~ | ~ w(z|_{N(j)})=2\} = \underset{j \in C} \bigcap \big\{ y \in ({\mathbb{R}_{\geq 0}})^{n} ~ | ~ \displaystyle\sum\limits_{i \in N(j) \setminus \{i_{0}\}} y_{i} \ge y_{i_{0}}, \forall i_{0} \in N(j) \big\}\\ 
&= \big\{ y \in ({\mathbb{R}_{\geq 0}})^{n} ~ | ~ \langle y,v_{i_{0},j} \rangle \ge 0 ~ \forall i_{0} \in N(j), ~ \forall j \in C \big\}
\end{align*}
where $v_{i_{0},j} \in \{-1,0,1\}^{n}$ is defined as follows: For all $i \in [n]$, \[
\big(v_{i_{0},j}\big)_{i} = \begin{dcases*}
        0  & if $i \notin N(j)$.\\ 
        -1 & if $i=i_{0}$.\\ 
	1 & if $i \in N(j) \setminus \{i_{0}\}$.
        \end{dcases*}
\]
Thus,
\begin{equation*}
\mathcal{K} = ({\mathbb{R}_{\geq 0}})^{n} \bigcap \underset{j \in C} \bigcap \big(cone\{ v_{i_{0},j} | i_{0} \in N(j) \}\big)^{*}= ({\mathbb{R}_{\geq 0}})^{n} \bigcap \underset{j \in C} \bigcap \big(D_{j}\big)^{*}
\end{equation*}
where for any $j \in C$, $D_{j} = cone\{ v_{i_{0},j} | i_{0} \in N(j) \}$. Note that if $L \subseteq \mathbb{R}^{n}$ is a cone, then its dual $L^{*}$ is also a cone. We will use below the following basic properties of dual cones:
\begin{enumerate}
\item[i)]\label{le:dual_cone_first_prop} If $L_{1},L_{2} \subseteq \mathbb{R}^{n}$ are cones, then $(L_{1}+L_{2})^{*} = L_{1}^{*} \cap L_{2}^{*}$.
\item[ii)]\label{le:dual_cone_second_prop} If $L \subseteq \mathbb{R}^{n}$ is a cone, then $(L^{*})^{*} = cl(L)$.
\end{enumerate}
Therefore, there is LP decoding success for $\eta$ on $\mathcal{K}$ if and only if $\widetilde{\eta} \in D$ where:
\begin{align*}
D &= int(\mathcal{K}^{*}) = int\bigg(\Big(({\mathbb{R}_{\geq 0}})^{n} \bigcap \underset{j \in C} \bigcap D_{j}^{*}\Big)^{*}\bigg) = int\bigg(\Big(\big(({\mathbb{R}_{\geq 0}})^{n}\big)^{*} \bigcap \underset{j \in C} \bigcap D_{j}^{*}\Big)^{*}\bigg) = int\bigg(\Big(\big(({\mathbb{R}_{\geq 0}})^{n} + \displaystyle\sum\limits_{j \in C} D_{j}\big)^{*}\Big)^{*}\bigg)
\end{align*}
and where the third equality follows from the fact that $({\mathbb{R}_{\geq 0}})^{n}$ is a self-dual cone and the last equality follows from property (i) above. Note that for any $j \in C, ~ D_{j}$ is a cone. Moreover, since $({\mathbb{R}_{\geq 0}})^{n}$ is a cone and the sum of any two cones is also a cone, it follows that $({\mathbb{R}_{\geq 0}})^{n} + \displaystyle\sum\limits_{j \in C} D_{j}$ is also a cone. Furthermore, by property (ii) above, we get that $D=int\bigg(cl\Big(({\mathbb{R}_{\geq 0}})^{n} + \displaystyle\sum\limits_{j \in C} D_{j}\Big)\bigg)$. Being a cone, $({\mathbb{R}_{\geq 0}})^{n} + \displaystyle\sum\limits_{j \in C} D_{j}$ is a convex set. For any convex set $S \subseteq \mathbb{R}^n$, we have that $int(cl(S)) = int(S)$ (See Lemma $5.28$ of \cite{aliprantis2006infinite}). Therefore,
\begin{align*}
D &= int\big(({\mathbb{R}_{\geq 0}})^{n} + \displaystyle\sum\limits_{j \in C} D_{j}\big)\\ 
&= ({\mathbb{R}_{> 0}})^{n} + \displaystyle\sum\limits_{j \in C} D_{j} \text{ (using Lemma \ref{le:topologyneededclaim} and the fact that $\displaystyle\sum\limits_{j \in C} D_{j}$ is a convex subset of $\mathbb{R}^{n}$)}\\ 
&=\{z \in \mathbb{R}^{n} ~ | ~ \exists y \in \displaystyle\sum\limits_{j \in C} D_{j} \text{ s.t. } z > y \}\\ 
&=\big\{z \in \mathbb{R}^{n} ~ | ~ \exists \{ \lambda_{i_{0},j}\}_{i_{0} \in N(j), j \in C} \text{ s.t. } \lambda_{i_{0},j} \geq 0 ~ \forall i_{0} \in N(j), \forall j \in C \text{ and } \displaystyle\sum\limits_{i_{0} \in N(j), j \in C} \lambda_{i_{0},j}v_{i_{0},j} < z \big\}\\ 
&=\big\{ \displaystyle\sum\limits_{i_{0} \in N(j), j \in C} \lambda_{i_{0},j}v_{i_{0},j} + u ~ | ~ \lambda_{i_{0},j} \geq 0 ~ \forall i_{0} \in N(j), \forall j \in C \text{ and } u \in ({\mathbb{R}_{> 0}})^{n}\big\}
\end{align*}
Thus, there is LP decoding success for $\eta$ on $\zeta$ if and only if there exist $\lambda_{i_{0},j} \geq 0$ for all $i_{0} \in N(j)$ and all $j \in C$ s.t. $\displaystyle\sum\limits_{i_{0} \in N(j), j \in C} \lambda_{i_{0},j}v_{i_{0},j} < \widetilde{\eta}$. Let $w(i,j) = \big(\displaystyle\sum\limits_{i_{0} \in N(j)} \lambda_{i_{0},j}v_{i_{0},j}\big)_{i} \text{ for all } i \in [n] \text{ and all } j \in C$. Since $(v_{i_{0},j})_{i} = 0$ whenever $i \notin N(j)$, we have that for every $i \in [n]$:
\begin{equation*}
\displaystyle\sum\limits_{j \in N(i)} w(i,j) = \displaystyle\sum\limits_{j \in N(i)} \big(\displaystyle\sum\limits_{i_{0} \in N(j)} \lambda_{i_{0},j}v_{i_{0},j}\big)_{i} = \displaystyle\sum\limits_{j \in C} \big(\displaystyle\sum\limits_{i_{0} \in N(j)} \lambda_{i_{0},j}v_{i_{0},j}\big)_{i} = \big(\displaystyle\sum\limits_{i_{0} \in N(j), j \in C} \lambda_{i_{0},j}v_{i_{0},j}\big)_{i} < \widetilde{\eta}_{i}
\end{equation*}
Moreover, for all $j \in C, ~ i_{1}, i_{2} \in N(j) \text{ s.t. } i_{1} \neq i_{2}$, we have
\begin{equation*}
w(i_{1},j) + w(i_{2},j) = \displaystyle\sum\limits_{i_{0} \in N(j)} \lambda_{i_{0},j} \Big( \big(v_{i_{0},j}\big)_{i_{1}} + \big(v_{i_{0},j}\big)_{i_{2}} \Big) \geq 0
\end{equation*}
since $\big(v_{i_{0},j}\big)_{i_{1}} + \big(v_{i_{0},j}\big)_{i_{2}} \geq 0$ because $i_{1} \neq i_{2} \in N(j)$.
We conclude that LP decoding success for $\eta$ on $\zeta$ is equivalent to the existence of a dual witness for $\eta$ on $\zeta$.
\end{proof}

\subsection{Proof of Lemmas \ref{le:finalupperbound} and \ref{le:upperboundcalc2}}\label{le:proof_unified}
The goal of this section is prove the following theorem which is used in the proofs of Lemmas \ref{le:finalupperbound} and \ref{le:upperboundcalc2}.
\begin{theorem}\label{le:unified_theorem}
Let $\lambda, \beta,m$ be positive integers with $\beta > d_{c}-1$ and $m \geq \lambda$. Consider the optimization problem:
\begin{equation}\label{le:max_equation_unified} v^{*} = \max_{ (T_{0}, \dots , T_{h}) \in W_h \atop h \in \mathbb{N}, h \geq 1} { f(T_{0}, \dots , T_{h}) }\end{equation}
where: 
\begin{equation*} f(T_{0}, \dots , T_{h}) = \displaystyle\sum\limits_{i=0}^{h} \frac{T_{i}}{(d_{c}-1)^{i}} \end{equation*}
and $W_h$ is the set of all tuples $(T_{0}, \dots , T_{h}) \in \mathbb{N}^{h+1}$ satisfying the following three equations:
\begin{equation}\label{le:sum_equation_unified}
\displaystyle\sum\limits_{i=0}^{h} T_{i} = m
\end{equation}
\begin{equation}\label{le:base_equation_unified}
T_{0} \le \lambda
\end{equation}
\begin{equation}\label{le:inductive_equation_unified}
T_{i+1} \le \beta T_{i}\text{ for all } i \in \{0,\dots,h-1\}
\end{equation}
Then,
\begin{equation*}
v^{*} \le \lambda \frac{\big(\frac{\beta}{d_{c}-1}\big)^2}{\frac{\beta}{d_{c}-1}-1} m^{\frac{\ln{\beta}-\ln(d_{c}-1)}{\ln{\beta}}}
\end{equation*}
\end{theorem}

\noindent We will first prove some lemmas which will lead to Lemma \ref{le:unified_theorem}.
\begin{definition}\label{le:firstbeta}
Let $l = \lfloor \log_{\beta}(\frac{m(\beta-1)}{\lambda}+1)\rfloor -1$.
\end{definition}
\noindent Note that $l \geq 0$ since $m \geq \lambda$.
\begin{lemma}\label{le:upperinductivebound}
Let \ensuremath{(T_{0}, \dots , T_{h}) \in W_h}. Then, $T_{i} \le \lambda \beta^{i}$ for all $i \in \{0,\dots,h\}$.
\end{lemma}

\begin{proof}[\bf{Proof of Lemma \ref{le:upperinductivebound}}]
Follows from equations (\ref{le:base_equation_unified}) and (\ref{le:inductive_equation_unified}).
\comments{
By induction on $i$:\\ 
Base case: $i=1$:
$T_{1} \le (d_{c}-1)$ by the definition of $W$.\\ 
Inductive step: We need to show that if the proposition holds for some $1 \le i \le {h-1}$, then it holds for $i+1$.\\ 
Assume that $T_{i} \le (d_{c}-1)\beta^{i-1}$. Then:
\begin{align*} T_{i+1} &\le (d_{c}-1)(d_{v}-1)T_{i} \text{ (by the definition of $W$)}\\ 
&\le (d_{c}-1)(d_{v}-1) (d_{c}-1)\beta^{i-1} \text{ (by the induction assumption)}\\ 
&= (d_{c}-1)\beta^{i}
\end{align*}
}
\end{proof}

\begin{lemma}\label{le:firstinW}
Let
\begin{align*}
&T_{i}'=\lambda \beta^{i} \text{ for all } i \in \{0,\dots,l\}\\ 
&T_{l+1}'= m - \lambda\frac{(\beta^{l+1} -1)}{(\beta - 1)}
\end{align*}
Then, $(T_{0}', \dots , T_{l+1}') \in W_{l+1}$.
\end{lemma}

\begin{proof}[\bf{Proof of Lemma \ref{le:firstinW}}]
First, note that \begin{math} (T_{0}', \dots , T_{l+1}') \in \mathbb{N}^{l+2} \end{math} since $T_{l+1}' \geq 0$ by Definition \ref{le:firstbeta}. Moreover,
\begin{align*} \displaystyle\sum\limits_{i=0}^{l+1} T_{i}' = \displaystyle\sum\limits_{i=0}^l \lambda \beta^{i} + T_{l+1}'
= \lambda \frac{(\beta^{l+1} -1)}{(\beta - 1)} + T_{l+1}'
= m
\end{align*}
We have that $T_{0}' \le \lambda$ and for every $i \in \{0,\dots,l-1\}$, $T_{i+1}' \le \beta T_{i}'$. We still need to show that $T_{l+1}' \le \beta T_{l}'$. We proceed by contradiction. Assume that $T_{l+1}' > \beta T_{l}'$. Then, $T_{l+1}' > \lambda \beta^{l+1}$. Thus,
\begin{equation*}
m = \displaystyle\sum\limits_{i=0}^{l+1} T_{i}' > \displaystyle\sum\limits_{i=0}^{l+1} \lambda \beta^{i} = \lambda \frac{(\beta^{l+2} -1)}{(\beta - 1)} > \lambda \frac{( \frac{m(\beta-1)}{\lambda}+1) -1}{(\beta - 1)} = m
\end{equation*}
since $l+2 = \lfloor \log_{\beta}(\frac{m(\beta-1)}{\lambda}+1) \rfloor + 1 > \log_{\beta}(\frac{m(\beta-1)}{\lambda}+1)$.
\end{proof}

\begin{lemma}\label{le:uniquesolutionRHS}
\ensuremath{(T_{0}', \dots , T_{l+1}')} is the unique (up to leading zeros) element that achieves the maximum in Equation (\ref{le:max_equation_unified}).
\end{lemma}

\begin{proof}[\bf{Proof of Lemma \ref{le:uniquesolutionRHS}}]
By Lemma \ref{le:firstinW}, \ensuremath{(T_{0}', \dots , T_{l+1}') \in W_{l+1}}. Let \ensuremath{(T_{0}, \dots , T_{h})\in W_{h}} such that $(T_{0}, \dots , T_{h})$ and $(T_{0}', \dots , T_{h}')$ are not equal up to leading zeros and without loss of generality assume that $h \geq l+1$ by extending $T$ with zeros if needed. In order to show that \ensuremath{ f(T_{0}, \dots , T_{h}) < f(T_{0}', \dots , T_{h}') }, we distinguish two cases:\\ 
Case 1: $(T_{0}, \dots , T_{l}) \neq (T_{0}', \dots , T_{l}')$. By Lemma \ref{le:upperinductivebound}, there exists \ensuremath{k_{1} \in \{0,\dots,l\} \text{ such that } T_{k_{1}} < \lambda \beta^{k_{1}}}. Therefore, \ensuremath{ \displaystyle\sum\limits_{i=0}^l T_{i}' - \displaystyle\sum\limits_{i=0}^l T_{i} > 0}. Note that:
\begin{align*}
f(T_{0}, \dots , T_{h}) - f(T_{0}', \dots , T_{l+1}') &= \displaystyle\sum\limits_{i=0}^l \frac{T_{i} - T_{i}'}{(d_{c}-1)^i} + \frac{T_{l+1} - T_{l+1}'}{(d_{c}-1)^{l+1}} + \displaystyle\sum\limits_{i=l+2}^h \frac{T_{i}}{(d_{c}-1)^i}\\ 
&\le \frac{1}{(d_{c}-1)^l} \displaystyle\sum\limits_{i=0}^l (T_{i} - T_{i}') + \frac{T_{l+1} - T_{l+1}'}{(d_{c}-1)^{l+1}} + \frac{1}{(d_{c}-1)^{l+1}}\displaystyle\sum\limits_{i=l+2}^h T_{i}\\ 
&= \frac{1}{(d_{c}-1)^l} \displaystyle\sum\limits_{i=0}^l (T_{i} - T_{i}') + \frac{1}{(d_{c}-1)^{l+1}}(\displaystyle\sum\limits_{i=l+1}^h T_{i}-T_{l+1}')\\ 
&= \frac{1}{(d_{c}-1)^l} \displaystyle\sum\limits_{i=0}^l (T_{i} - T_{i}') + \frac{1}{(d_{c}-1)^{l+1}}\displaystyle\sum\limits_{i=0}^l (T_{i}'-T_{i})
\end{align*}
Consequently, \begin{align*} f(T_{0}, \dots , T_{h}) &\le f(T_{0}', \dots , T_{l+1}') - \frac{( \displaystyle\sum\limits_{i=0}^l T_{i}' - \displaystyle\sum\limits_{i=0}^l T_{i} )}{(d_{c}-1)^l} + \frac{( \displaystyle\sum\limits_{i=0}^l T_{i}' - \displaystyle\sum\limits_{i=0}^l T_{i} )}{(d_{c}-1)^{l+1}}\\ 
&= f(T_{0}', \dots , T_{l+1}') - (d_{c}-2)\frac{( \displaystyle\sum\limits_{i=0}^l T_{i}' - \displaystyle\sum\limits_{i=0}^l T_{i} )}{(d_{c}-1)^{l+1}}\\ 
&< f(T_{0}', \dots , T_{l+1}') \end{align*}\\ 
Case 2: $(T_{0}, \dots , T_{l}) = (T_{0}', \dots , T_{l}')$. Then, $T_{l+1} \neq T_{l+1}'$. Since $T_{l+1}' = \displaystyle\sum\limits_{i=l+1}^{h} T_{i}$, we should have $T_{l+1}' - T_{l+1} > 0$. We have that
\begin{align*}
f(T_{0}, \dots , T_{h}) - f(T_{0}', \dots , T_{l+1}') &= \frac{T_{l+1}-T_{l+1}'}{(d_{c}-1)^{l+1}} + \displaystyle\sum\limits_{i=l+2}^{h} \frac{T_{i}}{(d_{c}-1)^i}\\ 
&\le \frac{T_{l+1}-T_{l+1}'}{(d_{c}-1)^{l+1}} + \frac{1}{(d_{c}-1)^{l+2}}\displaystyle\sum\limits_{i=l+2}^{h} T_{i}\\ 
&= \frac{T_{l+1}-T_{l+1}'}{(d_{c}-1)^{l+1}} + \frac{1}{(d_{c}-1)^{l+2}}\displaystyle\sum\limits_{i=0}^{l+1} (T_{i}'-T_{i})\\ 
&\le \frac{T_{l+1}-T_{l+1}'}{(d_{c}-1)^{l+1}} + \frac{(T_{l+1}'-T_{l+1})}{(d_{c}-1)^{l+2}}
\end{align*}
Consequently,
\begin{align*} f(T_{0}, \dots , T_{h}) &\le f(T_{0}', \dots , T_{l+1}') - \frac{( T_{l+1}' - T_{l+1} )}{(d_{c}-1)^{l+1}} + \frac{( T_{l+1}' - T_{l+1} )}{(d_{c}-1)^{l+2}}\\ 
&= f(T_{0}', \dots , T_{l+1}') - (d_{c}-2)\frac{( T_{l+1}' - T_{l+1} )}{(d_{c}-1)^{l+2}}\\ 
&< f(T_{0}', \dots , T_{l+1}') \end{align*}
\end{proof}

\begin{proof}[\bf{Proof of Lemma \ref{le:unified_theorem}}]
Let $\nu = \beta/(d_{c}-1)$. By Lemmas \ref{le:uniquesolutionRHS} and \ref{le:upperinductivebound}, we have that
\begin{align*} v^{*} \le \displaystyle\sum\limits_{i=0}^{l+1} \frac{T_{i}'}{(d_{c}-1)^i} &\le \displaystyle\sum\limits_{i=0}^{l+1} \lambda \frac{\beta^{i}}{(d_{c}-1)^i} = \lambda \displaystyle\sum\limits_{i=0}^{l+1} \nu^{i} = \lambda \frac{\nu^{l+2}-1}{\nu-1} < \lambda \frac{\nu^{l+2}}{\nu-1}\\ 
&\le \lambda \frac{\nu^{\log_{\beta}(\frac{m(\beta-1)}{\lambda}+1)+1}}{\nu-1} \le \lambda \frac{\nu^{2}}{\nu-1}\nu^{\log_{\beta}{m}}\\ 
&\le \lambda \frac{\nu^{2}}{\nu-1} m^{\frac{\ln{\nu}}{\ln{\beta}}} \end{align*}
\end{proof}

\subsection{Proof of Theorem \ref{le:asymptotictightness}}\label{le:app_asymptotic_tightness}
The goal of this section is to prove Theorem \ref{le:asymptotictightness} which is restated below.

\begin{reptheorem}{le:asymptotictightness} (Asymptotic tightness of Theorem \ref{le:maxweight} for $(d_{v},d_{c})$-regular LDPC codes)\\ 
There exists an infinite family of $(d_{v},d_{c})$-regular Tanner graphs $\{(V_{n},C_{n},E_{n})\}_n$, an infinite family of error patterns $\{\gamma_{n}\}_n$ and a positive constant $c$ s.t. there exists a hyperflow for $\gamma_{n}$ on $(V_{n},C_{n},E_{n})$ and any WDAG $(V_{n},C_{n},E_{n},w,\gamma_{n})$ corresponding to a hyperflow for $\gamma_{n}$ on $(V_{n},C_{n},E_{n})$ must have
\begin{equation*}
\underset{e \in E_{n}}{\operatorname{max}}{ |w(e)| } \geq c n^{\frac{\ln(d_{v}-1)}{\ln(d_{v}-1)+\ln(d_{c}-1)}}
\end{equation*}
\end{reptheorem}

We now prove some lemmas that lead to the proof of Theorem \ref{le:asymptotictightness}.
\begin{definition}\label{le:tight_construction_def} (Construction of $\{(V_{n},C_{n},E_{n})\}_n$)\\ 
Let $\beta = (d_{v}-1)(d_{c}-1)$. The Tanner graph $\{(V_{n},C_{n},E_{n})\}_n$ is constructed by connecting copies of the following two basic blocks:
\begin{enumerate}
\item The ``$A$ block'' $A_x$ with parameter the non-negative integer $x$. $A_x$ is an undirected complete tree rooted at a $(d_{v}-1)$-regular variable node. The internal nodes of $A_x$ other than the root are either $d_{c}$-regular check nodes or $d_{v}$-regular variable nodes. The leaves of $A_{x}$ are all $1$-regular variable nodes of depth $x$.\footnote{The depth of a variable node $v$ is the number of check nodes on the unique path from the root to $v$.} Thus, $A_x$ has $\beta^x$ leaves. An example $A$ block is given in Figure \ref{le:A_block_figure}.
\item The ``$B$ block'' $B_y$ with parameter the non-negative integer $y$. $B_y$ is an undirected tree rooted at a $(d_{v}-1)$-regular variable node. The internal nodes of $B_y$ other than the root are either $d_v$-regular variable nodes or $2$-regular check nodes. The leaves of $B_y$ are $1$-regular variable nodes. The nodes of $B_y$ are divided into $y+1$ layers indexed from $y$ to $0$. Layer $y$ consists of the root and the $(d_{v}-1)$ check nodes that are connected to the root. Each check node in layer $i$ is connected to a single variable node in layer $i-1$ for all $i=y,y-1,\dots,1$. Each variable node in layer $i$ is connected to $d_{v}-1$ check nodes in the same layer for all $i=y,y-1,\dots,1$. Thus, layer $0$ consists of $(d_{v}-1)^y$ leaves which are all $1$-regular variable nodes. An example $B$ block is given in Figure \ref{le:B_block_figure}.
\end{enumerate}
Let $\gamma = \frac{\ln(d_{v}-1)}{\ln(d_{v}-1)+\ln(d_{c}-1)}$. For every non-negative integer $n$, let $y_n = \lfloor \log_{(d_{v}-1)}{n^{\gamma}}\rfloor $ and $b_n = (d_{v}-1)^{y_n} = \Theta(n^{\gamma})$. The Tanner graph $\{(V_{n},C_{n},E_{n})\}_n$ is constructed using a root check node, one $B$ block, many $A$ blocks and some auxiliary variable and check nodes as follows:
\begin{enumerate}
\item Start with a check node $c_{0}$.
\item Connect $c_{0}$ to the roots of $d_{c}-1$ $A_{y_{n}+1}$ blocks and to the root of one $B_{y_n}$ block. Note that $B_{y_n}$ has $b_n$ leaves.
\item For every $i=y_{n},y_{n-1},\dots,1$, connect each check node in layer $i$ of $B_{y_n}$ to the roots of $(d_{c}-2)$ $A_{i}$ blocks. Note that there are $(d_v-1)^{y_n-i+1}$ check nodes in layer $i$.
\item\label{le:connecting_leaves} Let $T_n$ be the tree constructed so far and $l_n$ be its number of leaves. Note that all the leaves of $T_n$ are $1$-regular variable nodes. Complete $T_n$ into a $(d_v,d_c)$-regular graph by adding $O(l_n)$ $d_c$-regular new check nodes and (if needed) $O(l_n)$ $d_{v}$-regular new variable nodes in such a way that each new check is either connected to zero or to at least two leaves of the $B$ block.\footnote{Note that if $(d_v-1)l_n$ is divisible by $d_c$, we don't need any extra variable nodes. In the worst case, we can add $d_c$ copies of $T_n$ so that $(d_v-1) d_{c} l_{n}$ is divisible by $d_c$.}
\end{enumerate}
We call the check and variable nodes added in step \ref{le:connecting_leaves} the ``connecting'' check and variable nodes respectively.
\end{definition}

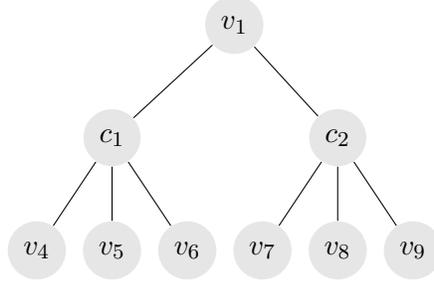
\begin{figure}[!h]

\centering

\begin{tikzpicture}
  [scale=.25,auto=left,every node/.style={circle,fill=gray!20}]
  \node (n2) at (11.5,16)  {$v_{1}$};
  \node (n5) at (5,10) {$c_{1}$};
  \node (n6) at (17,10)  {$c_{2}$};
  \node (n11) at (1,4)  {$v_{4}$};
  \node (n12) at (5,4)  {$v_{5}$};
  \node (n13) at (9,4)  {$v_{6}$};
  \node (n14) at (13,4)  {$v_{7}$};
  \node (n15) at (17,4)  {$v_{8}$};
  \node (n16) at (21,4)  {$v_{9}$};

  \foreach \from/\to in {n5/n2,n6/n2,n11/n5,n12/n5,n13/n5,n14/n6,n15/n6,n16/n6}
    \draw (\from) -- (\to);

\end{tikzpicture}

\caption{Example of an $A$ block with parameter $x=1$ where $d_{v}=3$ and $d_{c}=4$}\label{le:A_block_figure}

\end{figure}

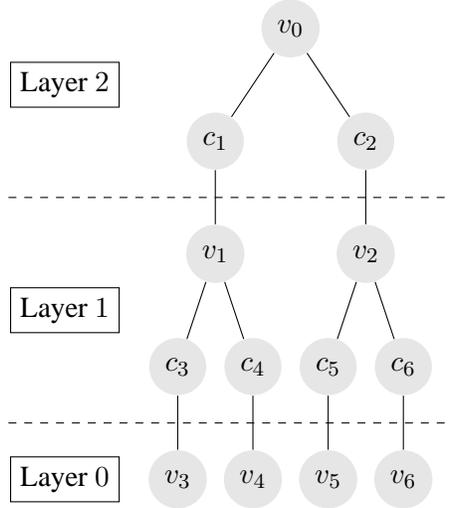
\begin{figure}[!h]

\centering

\begin{tikzpicture}
  [scale=.25,auto=left,every node/.style={circle,fill=gray!20}]
  \node (n1) at (10,28) {$v_{0}$};
  \node (n2) at (6,22)  {$c_{1}$};
  \node (n3) at (14,22)  {$c_{2}$};
  \node (n4) at (6,16)  {$v_{1}$};
  \node (n5) at (14,16) {$v_{2}$};
  \node (n6) at (4,10)  {$c_{3}$};
  \node (n7) at (8,10)  {$c_{4}$};
  \node (n8) at (12,10)  {$c_{5}$};
  \node (n9) at (16,10)  {$c_{6}$};
  \node (n10) at (4,4)  {$v_{3}$};
  \node (n11) at (8,4)  {$v_{4}$};
  \node (n12) at (12,4)  {$v_{5}$};
  \node (n13) at (16,4)  {$v_{6}$};
  \node[rectangle,draw,fill=white!20] at (-2,25) {Layer $2$};
  \node[rectangle,draw,fill=white!20] at (-2,13) {Layer $1$};
  \node[rectangle,draw,fill=white!20] at (-2,4) {Layer $0$};

  \draw[dashed] (-5,19) -- (19,19);
  \draw[dashed] (-5,7) -- (19,7);

  \foreach \from/\to in {n2/n1,n3/n1,n4/n2,n5/n3,n6/n4,n7/n4,n8/n5,n9/n5,n10/n6,n11/n7,n12/n8,n13/n9}
    \draw (\from) -- (\to);

\end{tikzpicture}

\caption{Example of a $B$ block with parameter $y=2$ where $d_{v}=3$}\label{le:B_block_figure}

\end{figure}

\begin{definition}\label{le:tight_llr_def} (Construction of $\{\gamma_{n}\}_n$)\\ 
Let $\{(V_{n},C_{n},E_{n})\}_n$ be the Tanner graph given in Definition \ref{le:tight_construction_def}. The error pattern $\gamma_{n}$ is defined by:
\begin{enumerate}
\item For every variable node $v$ in an $A$ block, $\gamma_{n}(v) = 1$.
\item For every variable node $v$ in the $B$ block, $\gamma_{n}(v) = -1$.
\item For every connecting variable node $v$, $\gamma_{n}(v) = 1$.
\end{enumerate}
\end{definition}

\begin{lemma}\label{le:tight_parameters_code} (Size of the code)\\ 
For any positive integer $n$, the Tanner graph $\{(V_{n},C_{n},E_{n})\}_n$ given in Definition \ref{le:tight_construction_def} is a $(d_{v},d_{c})$-regular code with $\Theta(n)$ variable nodes.
\end{lemma}

\begin{proof}[\bf{Proof of Lemma \ref{le:tight_parameters_code}}]
It is enough to show that the number $l_n$ of leaves of $T_n$ is $O(n)$. The number of leaves of block $B_{y_n}$ is $b_n = \Theta(n^{\gamma})$. The number of leaves of block $A_y$ is $(d_{v}-1)^y$. Thus, the number of leaves in all the $A$-blocks is
\begin{align*}
a_n &= (d_c-1)(d_v-1)^{y_n+1} + (d_c-2)\displaystyle\sum\limits_{i=1}^{y_n}(d_{v}-1)^{y_n-i+1} \beta^i\\ 
&= O((d_{v}-1)^{y_n}) + O((d_{v}-1)^{y_n}\displaystyle\sum\limits_{i=1}^{y_n}(d_{c}-1)^i)\\ 
&= O(b_n+\beta^{y_n})
\end{align*}
because $(d_v-1)^{y_n} = b_n$ and $\displaystyle\sum\limits_{i=1}^{y_n}(d_{v}-1)^i = O((d_{c}-1)^{y_n})$. Since $\beta^{y_n} = \Theta(n)$ and $b_n = o(n)$, we get that $l_n = b_n + a_n = \Theta(n)$.
\end{proof}

\begin{lemma}\label{le:tight_existence_hyperflow} (Existence of a hyperflow for $\{\gamma_{n}\}_n$ on $\{(V_{n},C_{n},E_{n})\}_n$)\\ 
Let $\{(V_{n},C_{n},E_{n})\}_n$ be the Tanner graph given in Definition \ref{le:tight_construction_def} and let $\gamma_{n}$ be the error pattern given in Definition \ref{le:tight_llr_def}. Then, for every positive integer $n$, there exists a hyperflow for $\gamma_{n}$ on $(V_{n},C_{n},E_{n})$.
\end{lemma}

\begin{proof}[\bf{Proof of Lemma \ref{le:tight_existence_hyperflow}}]
Let $\epsilon > 0$. We will further specify $\epsilon$ at the end of the proof. Consider the following assignment of weigths to edges of $E_{n}$:
\begin{enumerate}
\item In every $A$ block, the edges are directed toward the root of the block. The edges outgoing from the leaves have weight $1-\epsilon$. For every check node, the weight of the outgoing edge is equal to the common weight of its incoming edges. For each variable node, the sum of the weights of the outgoing edges is equal to the sum of the weights of the incoming edges plus $1-\epsilon$. Thus, the weight of the edge outgoing from the root of the $A_x$ block is
\begin{equation*}
r_x = (1-\epsilon)\displaystyle\sum\limits_{t=0}^{x}(d_{v}-1)^{t} = (1-\epsilon)\frac{(d_{v}-1)^{x+1}-1}{d_{v}-2}
\end{equation*}
\item In the $B$ block, the edges are directed toward the leaves. The edge connecting $c_0$ to the root of block $B$ has weight $w_{y_n}$ where for any $i \in \{0,\dots,y_n\}$:
\begin{equation*}
w_i \defeq (1+\epsilon) \displaystyle\sum\limits_{j=0}^{i}(d_{v}-1)^{j} = (1+\epsilon)\frac{(d_{v}-1)^{i+1}-1}{d_{v}-2}
\end{equation*}
For every internal variable node $v$, the weight of each outgoing edge from $v$ is $\frac{z-(1+\epsilon)}{d_{v}-1}$ where $z$ is the weight of the edge incoming to $v$. For every internal check node $c$, the weight of the edge outgoing from $c$ is equal to the weight of the edge incoming to $c$. By induction on the layer index $i=y_n,y_{n-1},\dots,0$, for every variable node $v$ in layer $i$, the weight of its incoming edge is $w_i$ and (if $v$ is not a leaf) the weight of each of its outgoing edges is $w_{i-1}$ (since $w_i$ satisfies the recurrence $w_{i-1} = \frac{w_i-(1+\epsilon)}{d_{v}-1}$ for all $i = y_n,y_{n-1},\dots,1$). 
\item All edges adjacent to connecting check or variable nodes have weight zero.
\end{enumerate}
By construction, the weights satisfy the dual witness equations (\ref{le:dw_var_equation}) and (\ref{le:dw_check_equation}) for all check and variable nodes in $A$ blocks, all internal variable nodes in the $B$ block and all the connecting check and variable nodes. To guarantee that equations (\ref{le:dw_var_equation}) and (\ref{le:dw_check_equation}) hold for the root check node $c_0$, we need that $r_{y_{n+1}} \geq w_{y_n}$. To guarantee them for the internal check nodes of the $B$ block, we need that $r_{i+1} \geq w_{i}$ for all $i = y_{n}-1,\dots,1$. To guarantee them for the leaves of the $B$ block, we need that $w_{0}-1 > 0$, which holds since $w_{0} = 1+\epsilon$. Thus, for every $i=y_{n},y_{n-1},\dots,1$, we need that $r_{i+1} \geq w_i$, i.e.
\begin{equation*}
(1-\epsilon)\frac{(d_{v}-1)^{i+2}-1}{d_{v}-2} \geq (1+\epsilon)\frac{(d_{v}-1)^{i+1}-1}{d_{v}-2}
\end{equation*}
which can be guaranteed by letting $0<\epsilon<1-\frac{2}{d_v}$.
\end{proof}

\begin{lemma}\label{le:lower_bound_any_hyperflow} (Lower bound for any hyperflow for $\{\gamma_{n}\}_n$ on $\{(V_{n},C_{n},E_{n})\}_n$)\\ 
For any positive integer $n$, any WDAG $(V_{n},C_{n},E_{n},w,\gamma_{n})$ corresponding to a hyperflow for $\gamma_{n}$ on $(V_{n},C_{n},E_{n})$ must have
\begin{equation*}
\underset{e \in E_{n}}{\operatorname{max}}{ |w(e)| } \geq c n^{\frac{\ln(d_{v}-1)}{\ln(d_{v}-1)+\ln(d_{c}-1)}}
\end{equation*}
for some constant $c>0$.
\end{lemma}

\begin{proof}[\bf{Proof of Lemma \ref{le:lower_bound_any_hyperflow}}]
Let $(V_{n},C_{n},E_{n},w,\gamma_{n})$ be a WDAG corresponding to a hyperflow for $\gamma_{n}$ on $(V_{n},C_{n},E_{n})$. Since $\gamma_n(v) = -1$ for every leaf $v$ of the $B$ block (which has $b_n$ leaves) and since each connecting check node adjacent to a leaf of the $B$ block is connected to at least two leaves of the $B$ block, there should be a flow of total value larger than $b_n$ from the non-leaf and non-connecting nodes of the $B$ block to its leaves. Applying the same argument inductively and using the fact that for every variable node $v$ of the $B$ block $\gamma_n(v)=-1$, we get that all the edges of the $B$ block should be oriented toward its leaves and that there should be a flow of value larger than $b_n$ entering the root of the $B$ block. Thus, the edge connecting $c_{0}$ to the root of the $B$ block should be oriented toward the $B$ block and should have value larger than $b_n = \Theta(n^{\frac{\ln(d_{v}-1)}{\ln(d_{v}-1)+\ln(d_{c}-1)}})$.
\end{proof}

\begin{proof}[\bf{Proof of Theorem \ref{le:asymptotictightness}}]
Follows from Lemmas \ref{le:tight_parameters_code}, \ref{le:tight_existence_hyperflow} and \ref{le:lower_bound_any_hyperflow}.
\end{proof}

\bibliographystyle{alpha}
\bibliography{LP_scc_1.bib}

\end{document}